\documentclass[11pt,reqno]{amsart}  

\usepackage{amscd} 
\usepackage{mathrsfs}
\usepackage{amsmath,amsbsy,amsthm,amsfonts,amssymb,esint}
\usepackage{mathtools}
\usepackage[unicode,psdextra]{hyperref}
\usepackage{graphicx,subcaption}
\usepackage{tikz}
\usepackage{bm}
\usepackage[normalem]{ulem}
\usepackage{soul}
\usepackage{color}

\usepackage[sort&compress,capitalise,nameinlink]{cleveref}
\crefname{section}{\textsection}{\textsection}
\crefname{subsection}{\textsection}{\textsection}
\usepackage{dsfont}
\usepackage{enumerate}

\DeclareMathAlphabet\mathbfcal{OMS}{cmsy}{b}{n}

\newtheorem{theorem}{Theorem}[section]
\newtheorem{corollary}[theorem]{Corollary}
\newtheorem{lemma}[theorem]{Lemma}
\newtheorem{proposition}[theorem]{Proposition}
\theoremstyle{definition}
\newtheorem{remark}[theorem]{Remark}

\numberwithin{equation}{section}

\numberwithin{equation}{section}

\AtEndEnvironment{remark}{\hfill $\diamond$}

\newcommand{\N}{\mathbb{N}}
\newcommand{\Nz}{\mathbb{N}_0}
\newcommand{\Z}{\mathbb{Z}}

\newcommand{\R}{\mathbb{R}}
\newcommand{\C}{\mathbb{C}}

\newcommand{\one}{\mathds{1}}

\newcommand{\tx}{\textstyle}

\def\LN{L_{\mbox{{\tiny{$N$}}}}}
\def\Lz{L_{0}}
\def\AN{A_{\beta,N}}
\def\Az{A_{\beta,0}}
\def\fe{\phi_{\mbox{{\tiny{$E$}}}}}
\def\ZeN{\mathcal{Z}_{\beta,N}}
\def\Zez{\mathcal{Z}_{\beta,0}}

\def\VN{V_{N}^{z}}
\def\PN{P_{N}^{z}}
\def\QN{Q_{N}^{z}}
\def\SN{S_{N}^{z}}
\def\QNp{Q_N^{+,v^2}}
\def\PNp{P_N^{+,v^2}}
\def\SNp{S_N^{+,v^2}}
\def\VNp{V_{N}^{+,v^2}}
\def\PNp{P_{N}^{+,v^2}}

\def\II{\mathcal{I}}
\def\KK{\Xi}

\def\H0{H_{0}}

\def\Gz{G_{0}}

\def\GN{G_{N}}
\def\GaN{\Gamma_N}
\def\rr{\mathrm{r}}
\def\ee{\mathrm{e}}

\def\dom{\mathscr{D}}

\def\mm{m}
\def\nn{n}
\def\tt{t}

\def\IR{\mathrm{IR}}
\def\UV{\mathrm{UV}}
\def\JIR{J_\IR}
\def\JUV{J_\UV}
\def\ren{\mathrm{ren}}
\def\vac{\mathrm{vac}}
\def\vz{v_0}
\def\Fr{\mathcal{F}_{\ren}}
\def\Ur{\mathcal{U}_{\ren}}
\def\Sr{\mathcal{S}_{\ren}}

\renewcommand{\Re}{\mathrm{Re}}
\renewcommand{\Im}{\mathrm{Im}}
\renewcommand{\dots}{...\,}

\DeclareMathOperator*{\Res}{Res}
\DeclareMathOperator*{\Rez}{Res_0}
\DeclareMathOperator*{\Reu}{Res_1}
\DeclareMathOperator*{\Rek}{Res_k}
\DeclareMathOperator*{\Tr}{Tr}

\title[Vacuum and thermal fluctuations of a scalar field with point interactions]{Vacuum and thermal fluctuations of a scalar field\\ with point interactions}

\author{Davide Fermi}
\address{Dipartimento di Matematica, Politecnico di Milano, P.zza Leonardo da Vinci, 32, 20133, Milano, Italy\\
and Istituto Nazionale di Fisica Nucleare, Sezione di Milano, Italy}
\email{davide.fermi@polimi.it}
\urladdr{https://fermidavide.com}

\author{Marco Gurgoglione}
\address{Dipartimento di Matematica, Politecnico di Milano, P.zza Leonardo da Vinci, 32, 20133, Milano, Italy}
\email{marco.gurgoglione@mail.polimi.it}

\begin{document}

\begin{abstract}
We investigate the vacuum and thermal fluctuations of a neutral massless scalar field living in Minkowski spacetime and interacting with a finite number of point-like obstacles, modelled by zero-range potentials. The system is described rigorously in terms of self-adjoint realizations of the Laplacian, under assumptions ensuring the absence of instabilities.
Using the relative zeta-function technique, we determine the renormalized connected partition function and derive explicit expressions for the thermodynamic observables, characterizing both their low- and high-temperature behaviors.
Furthermore, we derive of a convergent Born series expansion for the Casimir energy, which identifies multiple-scattering processes as the mechanism underlying vacuum forces. The latter decompose into pairwise contributions directed along the lines joining the obstacles, with intensities depending non-locally on the full configuration. We also present some numerical results for identical obstacles, indicating that the Casimir forces are always attractive in this context.	
\end{abstract}


\maketitle

\vspace{-0.5cm}
\tableofcontents
\vspace{-0.8cm}

\section{Introduction}

The Casimir effect is a force arising between uncharged conducting bodies, whose origin can be ascribed to quantum vacuum fluctuations of the electromagnetic field. First predicted by Casimir \cite{C48} and rooted in earlier work with Polder on retarded van der Waals interactions \cite{CP48}, it was soon framed within a broader class of dispersion forces by Lifshitz \cite{L55}.
Increasingly precise measurements \cite{L97,MR98,BC+02,DM24} have since confirmed these predictions and stimulated theoretical studies on the role of geometry, material properties and temperature \cite{M94,MT97,M01,L05,BK+09,DM+11,FP17}.
In addition to its fundamental significance, Casimir physics has recently gained technological relevance. At separations below 100nm, dispersion forces dominate interactions between neutral surfaces, causing stiction in Micro and Nano Electro Mechanical Systems (MEMS/NEMS) \cite{EY+24,TP25} and representing an unavoidable background in precision measurements \cite{KC+07}.
Current research explores the engineering of repulsive Casimir forces for frictionless levitation \cite{DL+14,SC+23} and exploiting critical interactions to control colloidal dispersions \cite{NF+13,PR+25}.

In this work we study an exactly solvable model allowing for an explicit and rigorous analysis of some fundamental features of Casimir interactions. We examine the vacuum and thermal fluctuations of a neutral massless scalar field living in Minkowski space-time and interacting with fixed pointlike obstacles, modeled by means of zero-range potentials in the spatial part of the Klein-Gordon equation which governs the field dynamics. Working in natural units with $\hbar = 1$ and $c= 1$, we assume the field to satisfy
\begin{equation*}
	(\partial_{tt} + \LN)\, \phi = 0\,,
\end{equation*}
where $t \in \R$ is the time coordinate and $\LN$ is the differential operator formally given by
\begin{equation}\label{eq:Lheu}
	\mbox{``}\,\LN = -\Delta + \tx\sum_{\nn = 1}^N \mu_\nn\,\delta_{x_\nn} \mbox{''} .
\end{equation}
Here, $-\Delta$ is the usual Laplacian on $\R^3$, $N \in \N$ is the number of obstacles and, for each $\nn \in \{1,\dots,N\}$, $\delta_{x_\nn}$ is the Dirac delta distribution concentrated at the point $x_\nn \in \R^3$, while $\mu_\nn$ is the corresponding (infinitesimal) coupling strength. We will refer to the standard realization of the heuristic expression \eqref{eq:Lheu} as a self-adjoint extension of the symmetric operator $-\Delta\!\upharpoonright\! C^\infty_c(\R^3 \setminus \{x_1,\dots,x_N\})$ acting in $L^2(\R^3)$ \cite{AG+88,AK99}, making suitable assumptions on the interaction strengths and on the distances between the points to ensure the stability of the field theory. 
Zero-range potentials have been used in the physics literature since the early days of quantum mechanics \cite{BP34,F36}. In the present context, they are meant to provide an effective description of regimes in which the typical size of extended obstacles is much smaller than the characteristic wavelength of the quantum field. 

Treating the obstacles as a fixed background and neglecting back-reaction, thermal effects at equilibrium can be accounted for using Matsubara's imaginary time formalism \cite{G77,DK78,KG89,BC+03}. Introducing the Wick-rotated coordinate $\tau = i t$ and imposing the KMS periodicity identification $\tau \sim \tau + \beta$, with $\beta > 0$ denoting the inverse temperature, the Euclidean canonical partition function for the system under analysis is formally given by the path-integral
\begin{equation}\label{eq:ZeN}
	\mbox{``}\,\ZeN = \int \mathcal{D}\fe \; e^{-\frac{1}{2} \langle \fe| \AN \fe \rangle}\;\mbox{''},
\end{equation}
where the relevant operator is $\AN = -\partial_{\tau\tau} + \LN$, acting in $L^2(S^1_{\beta/2\pi} \!\times \R^3)$. 
Of course, the expression \eqref{eq:ZeN} is ill-defined, due to ultraviolet divergences. Physical information can nonetheless be extracted by isolating the contribution of the obstacles. To this purpose, one introduces an ultraviolet regularization and compares $\ZeN$ with the free partition function $\Zez$, corresponding to the configuration without obstacles ($N=0$). Subtracting the free contribution and finally removing the regularization yields a finite and physically meaningful result. 

In this spirit, handling \eqref{eq:ZeN} as if it were a finite-dimensional Gaussian integral, by heuristic manipulations one gets
\begin{multline*}
	\ln \ZeN - \ln \Zez 
	= - \tfrac{1}{2} \log \det(\ell^2 \AN) + \tfrac{1}{2} \log \det(\ell^2 \Az) \\
	= - \tfrac{1}{2} \Tr \log(\ell^2 \AN) + \tfrac{1}{2} \Tr \log(\ell^2 \Az)
	= \tfrac{1}{2} \Tr \!\left.\tfrac{d}{ds}\right|_{s = 0} (\ell^2 \AN)^{-s} - \tfrac{1}{2} \Tr\! \left.\tfrac{d}{ds}\right|_{s = 0}(\ell^2 \Az)^{-s} \\
	= \Big[\tfrac{1}{2} \tfrac{d}{ds} \Tr\! \big[ \AN^{-s} - \Az^{-s} \big]
		- (\log\ell)\, \Tr \!\big[ \AN^{-s} - \Az^{-s} \big] \Big]_{s = 0} \,,
\end{multline*}
where $\ell> 0$ is a length scale, introduced for dimensional reasons, and $s$ is a complex parameter allowing to regularize the ultraviolet divergences. To proceed further, one can use the Euler integral representation of the Gamma function \cite[Eq. 5.2.1]{OL+10} to write
\begin{equation*}
	\Tr \!\big[ \AN^{-s} - \Az^{-s} \big] = \frac{1}{\Gamma(s)} \int_{0}^{\infty}\!\! du\;u^{s-1} \Tr \!\big[ e^{-u\, \AN} - e^{-u\, \Az} \big]\,.
\end{equation*}
Under suitable assumptions, this expression is well-defined for values of $s$ ranging in an appropriate strip of the complex plane and it provides a natural definition of the {\sl relative zeta function} $\zeta(s; \AN,\Az)$. This function admits, in turn, a meromorphic extension to the entire complex plane.

A rigorous implementation of the above arguments was provided by M\"uller \cite{M98}, highlighting deep connections with scattering theory. His approach can be viewed as a variation of the zeta-regularization method originally introduced in quantum field theory by Dowker and Critchley \cite{DC76}, Hawking \cite{H77}, and Wald \cite{W79}, building in turn on earlier ideas of Ray and Singer \cite{RS74} for determinants of elliptic operators. For reviews and more recent developments on zeta-function regularization, see, for instance, \cite{EO+94,K01,BC+03,FP17} and references therein. A key feature of M\"uller’s formulation is its ability to accommodate operators with continuous spectrum, as required in the present setting.
In this framework, one defines the renormalized connected relative partition function (or effective Euclidean action) as
\begin{equation}\label{eq:lnZN}
	\big[\ln \ZeN - \ln \Zez\big]_{\ren} := \Big[\tfrac{1}{2}\,\zeta'(s;\AN,\Az) - (\log\ell)\, \zeta(s;\AN,\Az) \Big]_{s = 0}\,,
\end{equation}
where the evaluation at $s = 0$ is understood in the sense of analytic continuation. 

An equivalent formulation of M\"uller's method, based on resolvent operators rather than heat semigroups, was proposed by Spreafico and Zerbini \cite{SZ09}. In that work, they also computed the renormalised relative partition function for a massless scalar field
in presence of one or two delta-like potentials ($N=1,2$), without addressing the attractive or repulsive nature of the resulting forces. The same formalism was later applied to models involving a point-like impurity and an infinite conducting plate \cite{AC+10}, a point-like impurity combined with a Coulomb potential centered at the same point \cite{ACS16}, and a semi-transparent hyperplane selecting transverse modes \cite{CFP17}.

The study of Casimir forces for a scalar field with singular potentials supported on sets of high co-dimension was initiated in the physics literature by Scardicchio \cite{S05}, with further developments appearing in \cite{GS17,GS25} (see also the references cited therein). A rigorous renormalization of both local and global observables for a single point impurity was subsequently achieved in \cite{FP18,F20,FP23}, building on a local formulation of zeta-function regularization grounded in functional analytic techniques developed in \cite{FP11,FP15,FP16,F16,FP17}. Analogous results for one and two centers were obtained in \cite{Z21} under an appropriate scaling limit by means of a refined version of the point-splitting procedure outlined in \cite{H05,H06,HS10}. 
We also take this opportunity to mention \cite{S21,FS22,S24} for another rigorous description of Casimir physics based on spectral theory and integral kernels methods, including results that establish the equivalence of different approaches.

We employ the approach of Spreafico and Zerbini \cite{SZ09}, extending it to configurations with an arbitrary finite number $N \geqslant 1$ of point-like obstacles. This generalization is far from trivial, as the expressions involved rapidly grow in complexity with $N$. Indeed, a direct extension of \cite{SZ09} would require inverting $N \times N$ matrices which, while theoretically possible, becomes intractable for large $N$ and obstructs the analysis of the meromorphic structure of the relative zeta function, a key tool in our analysis. To overcome this difficulty, we develop an alternative approach based on a separation of infrared and ultraviolet contributions, which further provides a cleaner physical picture. 

We begin by reviewing the rigorous definition of the spatial operator in the Klein–Gordon equation and by stating our working assumptions in \cref{sec:hypo}. We then compute the relative zeta function and describe its pole structure in \cref{sec:zeta}. Building on this analysis, in \cref{sec:Zren} we determine the renormalized connected relative partition function, from which the main thermodynamic quantities - free energy, internal energy, and entropy - can be derived. We determine both the low-temperature asymptotics and the leading high-temperature behavior of these quantities; the Casimir energy arises naturally as the only contribution surviving in the zero-temperature limit. In \cref{sec:Born} we provide a Born series expansion for the Casimir energy, discussing its physical interpretation. We specialize this expansion in \cref{sec:num} to the case of identical obstacles, presenting some numerical examples as well.
In \cref{sec:proofs} we collect the proofs, together with a number of auxiliary results. In \cref{subsec:restr} and \cref{subsec:respden} we evaluate the relative spectral trace and the relative spectral density, respectively. In \cref{subsec:acrelz} we establish the well-posedness of the relative zeta function in a suitable complex strip and derive its meromorphic extension to arbitrarily large strips. This is used in \cref{subsec:relpart} to compute the renormalised relative partition function. In \cref{subsec:ded} we analyze the relative Dedekind eta function and its asymptotic expansions for small and large argument, which govern the thermal corrections in the low- and high-temperature regimes. In \cref{sec:bornexp} we derive the Born series expansion for the Casimir energy, together with convergence estimates. Finally, in Appendix \ref{sec:app} we evaluate a class of integrals in terms of the exponential integral function, enabling faster numerical computations for identical obstacles.

Our analysis highlights some key features of Casimir physics.
The Born series reveals that fluctuation-induced forces originate from multi-scattering processes and admit a decomposition into pairwise contributions directed along the line joining each pair of obstacles, while their intensities retain a nonlocal dependence on the full configuration. The numerical investigations presented in \cref{sec:num} suggest that these forces are always attractive, although establishing this property analytically remains challenging due to their intrinsically nonlocal nature. Finally, we emphasize that, despite the appearance of the auxiliary renormalization length $\ell$ at intermediate stages, the resulting forces depend only on intrinsic physical parameters, namely the positions of the obstacles and their scattering lengths.

\section{Main results}\label{sec:main}

\subsection{The model and working hypotheses}\label{sec:hypo}
We begin by recalling the definition of the self-adjoint operator which identifies the spatial part of the Klein-Gordon operator for the model under analysis, providing a rigorous realization of the heuristic expression \eqref{eq:Lheu}. This is characterized as a singular perturbation of the free Laplacian
	\begin{equation*}
		\Lz = -\Delta\,, \qquad \dom(\Lz) = H^2(\R^3)\,,
	\end{equation*}
where $H^2(\R^3)$ is the usual Sobolev space of order $2$. Considering that $\sigma(\Lz) = \sigma_{\mathrm{ac}}(\Lz) = [0,+\infty)$, we will refer to the associated resolvent operator 
	\begin{equation*}
		R(z;\Lz) := (\Lz - z)^{-1} : L^2(\R^3) \to H^2(\R^3)\,, \qquad \mbox{for $z \in \C \setminus [0,+\infty)$}\,,
	\end{equation*}
which acts by convolution with the Green function
	\begin{equation}\label{eq:G0}
		\Gz^{z}(x) := \frac{e^{i \sqrt{z}\,|x|}}{4\pi|x|}\,.
	\end{equation}
Here and in the sequel, the branch of the square root is chosen so that $\Im \sqrt{z} > 0$.

For fixed $N \in \N$, let $\alpha_1,\dots,\alpha_N \in \R$ and consider the $N \times N$ matrix
	\begin{equation}\label{eq:GaN}
		\big[\GaN^{z}\big]_{\mm\nn} := \left(\alpha_\nn - \frac{i\sqrt{z}}{4\pi}\right)\delta_{\mm\nn} - \Gz^{z}(x_\mm-x_\nn)\,(1-\delta_{\mm\nn})\,, \qquad (m,n = 1,\dots,N)\,.
	\end{equation}
Throughout the paper, expressions of the form $f(\mm,\nn)(1-\delta_{\mm\nn})$ are understood to vanish for $\mm=\nn$, so that any divergence of $f(\mm,\nn)$ on the diagonal does not contribute.

Following \cite[\S II.1.1]{AG+88}, we define the self-adjoint operator
	\begin{eqnarray}
		& \hspace{-2.cm}\dom(\LN) = \big\{\psi = \phi_z + \tx\sum_{\mm=1}^N q_\mm\, \Gz^{z}(\cdot - x_\mm) \in L^2(\R^3) \;\,\big|\; \phi_z\!\in\! H^2(\R^3),\;\, q_1,\dots,q_N \!\in\! \C\,,\;\, \mbox{ and}\nonumber \\
		& \hspace{7.8cm} \phi_z(x_\mm) = \tx\sum_{\nn = 1}^N\! \big[\GaN^z\big]_{\mm \nn}\, q_\nn\; \mbox{ for }\, \mm = 1,\dots,N \big\} \,, \nonumber  \\
		& (\LN -z) \psi = (\Lz -z) \phi_z\,. \label{eq:HNdef}
	\end{eqnarray}
where $z \in \C \setminus [0,+\infty)$ is arbitrary and the decomposition of $\psi$ is unique.
The parameters $\alpha_\nn$ and the coefficients $\mu_\nn$ in \eqref{eq:Lheu} are related by $\mu_\nn = - c\,\varepsilon + \alpha_\nn\,\varepsilon^2$, with $\varepsilon \to 0^+$ and $c > 0$, see \cite[\S II.1.2 and App. H]{AG+88}.
Let us also mention that varying $\alpha_\nn \in \R \cup \{+\infty\}$ yields the whole family of so-called {\em local} self-adjoint extensions of the closable symmetric operator $-\Delta\!\upharpoonright\! C^\infty_c(\R^3 \setminus \{x_1,\dots,x_N\})$. Notably, the free operator $\Lz$, namely the Friedrichs extension, is recovered in the limit $\alpha_\nn \to +\infty$ for all $n\in\{1,\dots,N\}$.

The spectrum of $\LN$ is given by \cite[Thm. 1.1.4]{AG+88}
	\begin{equation}\label{eq:sigLN}
		\sigma(\LN) = \sigma_{\mathrm{ac}}(\LN) \cup \sigma_{\mathrm{pp}}(\LN)\,, \qquad
		\sigma_{\mathrm{ac}}(\LN) = [0,+\infty)\,, \qquad
		\sigma_{\mathrm{pp}}(\LN) = \{z \in \R_-\,|\, \det\GaN^z = 0\}\,,
	\end{equation}
and, for any $z \in \C \setminus \sigma(\LN)$, the resolvent operator
	\begin{equation*}
		R(z;\LN) := (\LN - z)^{-1} : L^2(\R^3) \to \dom(\LN)\,, 
	\end{equation*}
acts by integration against the kernel \cite[Eq. (1.1.33)]{AG+88}
	\begin{equation}\label{eq:GN}
		\GN^{z}(x,y) = \Gz^{z}(x-y) + \sum_{\mm,\nn=1}^N \big[\GaN^{z}\big]^{-1}_{\mm\nn}\, \Gz^{z}(x- x_\mm)\,\Gz^{z}(x_\nn- y)\,.
	\end{equation}
Here and in the sequel, $\big[\GaN^{z}\big]^{-1}_{\mm\nn}$ indicates the $(m,n)$-component of the inverse matrix $\big(\GaN^{z}\big)^{-1}$.
\medskip
	
We shall henceforth assume
	\begin{equation}\label{eq:ass1}
		\alpha_\nn > 0\,, \qquad \mbox{for all $\nn = 1,\dots,N$}\,, 
	\end{equation}
and
	\begin{equation}\label{eq:ass2}
		\sum_{\mm \neq \nn} \frac{1}{(4\pi \alpha_\mm)^2\, |x_\mm - x_\nn|^2} < 1 \,.
	\end{equation}
	
\begin{remark}
Condition \eqref{eq:ass1} implies, in particular, that the single-center scattering lengths $-(4\pi \alpha_\nn)^{-1}$ are negative definite \cite[Eq. (1.4.10)]{AG+88}.
In physical terms, \eqref{eq:ass2} indicates that, depending on their number, the point-like obstacles are well separated and the associated single-center interactions are sufficiently weak. This condition is not optimal, but it provides a simple sufficient criterion with a clear interpretation. As a matter of fact, the zero-range approximation itself becomes questionable when the interactions are strong and the obstacles lie too close to one another. 
On the other hand, it is worth noting that \eqref{eq:ass1} and \eqref{eq:ass2} are compatible with a so-called homogenization regime in which the intensities and mutual distances simultaneously tend to zero as the number of obstacles grows, following the scaling $|x_{\mm}-x_{\nn}| \sim N^{-1/3}$ and $\alpha_{\nn} \sim N$ as $N\to\infty$, while the total interaction strength remains finite \cite{CCF26}.
\end{remark}

\subsection{Meromorphic structure of the relative zeta function}\label{sec:zeta}
Having established the model of interest in rigorous terms and outlined our working hypothesis, we now retrace the steps outlined in \cite{SZ09} to evaluate the relative zeta function for the pair $(\LN,\Lz)$.

We first observe that $R(z;\LN)$ is a finite-rank perturbation of $R(z;\Lz)$, and define the {\sl relative resolvent trace}
	\begin{equation}
		\rr(z;\LN,\Lz) := \Tr \big[ R(z;\LN) - R(z;\Lz) \big]\,, \qquad \mbox{for all $z \in \C \setminus \big( \sigma(\LN) \cup \sigma(\Lz)\big)$}\,.
	\end{equation}

	\begin{proposition}\label{prop:rz}
		Assume \eqref{eq:ass1} and \eqref{eq:ass2}. Then,
			\begin{equation}\label{eq:sigLNp}
				\sigma(\LN) = \sigma_{\mathrm{ac}}(\LN) = [0,+\infty)\,.
			\end{equation}
		Moreover:
		\begin{enumerate}[i)]
			
			\item there exists a sequence $\{a_{j}\}_{j \in \Nz} \subset \C$ such that
				\begin{equation}\label{eq:asyrz0}
				 	\rr(z;\LN,\Lz) \sim \sum_{j = 0}^{+\infty} a_j \,z^{j/2 - 1/2}\,,
				 	\qquad \mbox{as $z \to 0$ in $\C \setminus [0,+\infty)$}\,;
				\end{equation}							
			
			\item there exists a sequence $\{b_{j}\}_{j \in \Nz} \subset \C$ such that
				\begin{equation}\label{eq:asyrzinf}
				 	\rr(z;\LN,\Lz) \sim \sum_{j = 0}^{+\infty} b_{j}\,z^{-j/2-1}\,,
				 	\qquad \mbox{as $z \to \infty$ in $\C \setminus [0,+\infty)$}\,.
				\end{equation}	
									
		\end{enumerate}
	\end{proposition}

	\begin{remark}
		The identities in \eqref{eq:sigLNp} show that, under the hypotheses \eqref{eq:ass1} and \eqref{eq:ass2}, the operator $\LN$ is positive semidefinite. This ensures, in turn, that the field theory under analysis is free of instabilities. The latter otherwise occur whenever the spatial part of the Klein-Gordon operator posses negative eigenvalues. 
		As anticipated, the conditions \eqref{eq:ass1} and \eqref{eq:ass2} are not strictly necessary for this property; see the spectral characterization in \eqref{eq:sigLN}.		
	\end{remark}

	\begin{remark}\label{rem:Lamtc}
		The limit $z \to \infty$ in \eqref{eq:asyrzinf} is understood, following \cite{SZ09}, with $z$ approaching infinity along any contour of the form $\Lambda_{\theta,-c} = \{z \in \C\,|\, \arg(z + c) = \theta \; \mbox{or}\; \arg(z + c) = 2\pi - \theta\}$, for fixed $c > 0$ and $0 < \theta < \pi/2$. The notation $f(z) \sim \sum_{j = 0}^{+\infty} c_j z^{\gamma_j}$ for $z \to z_*$, with $z_* = 0$ or $z_* = \infty$, indicates that $f(z) - \tx \sum_{j = 0}^{J} c_j\, z^{\gamma_j} = O(z^{\gamma_{J+1}})$ for any $J \in \Nz$, regardless of the convergence of the series. 
		Note also that our convention for resolvent operators differs by a sign from that adopted in \cite{SZ09}.
	\end{remark}

On account of \cref{prop:rz} and \cite[Lemma 2.1]{SZ09}, the pair of non-negative operators $(\LN,\Lz)$ satisfies conditions analogous to (1.1)-(1.3) of \cite{M98} for the associated semi-groups $(e^{-\tt \LN},e^{-\tt\Lz})$, with $\tt > 0$. Here $\tt$ is the standard symbol employed for the semi-groups' parameter and should not be confused with the time coordinate.
Accordingly, we can define the {\sl relative zeta function} by the equation 
	\begin{equation}\label{eq:zetamull}
		\zeta(s;\LN,\Lz) := \frac{1}{\Gamma(s)} \int_0^\infty\! d\tt\;\tt^{s-1} \Tr\big[ e^{-\tt \LN} - e^{-\tt \Lz}\big]\,,
	\end{equation}
for $s$ ranging in the complex strip 
	\begin{equation}\label{eq:sinstrip}
		\{ s \in \C\;|\; 0 < \Re s < 1/2\}\,,
	\end{equation}	
The function $\zeta(s;\LN,\Lz)$ is analytic in this strip and admits a meromorphic extension to the entire complex plane. For the explicit construction of the analytic continuation it is convenient to consider the equivalent representation
	\begin{equation}\label{eq:zetaev}
		\zeta(s;\LN,\Lz) = \int_0^{\infty} dv\;v^{-2s}\,\ee(v;\LN,\Lz)\,,
	\end{equation}
in terms of the {\sl relative spectral density}
	\begin{equation}\label{eq:defee}
		\ee(v;\LN,\Lz) := \frac{v}{i\pi} \lim_{\varepsilon \to 0^+} \big[ \,\rr(v^2 e^{i\varepsilon};\LN,\Lz) - \rr(v^2 e^{i(2\pi - \varepsilon)};\LN,\Lz) \,\big] \,.
	\end{equation}
We remark once again that the convention for the resolvent operator adopted here differs from that used in \cite{SZ09}, but the signs in the corresponding formulas have been adjusted accordingly.

	\begin{proposition}\label{prop:ee}
		Assume \eqref{eq:ass1} and \eqref{eq:ass2}. Then, the map $v \in (0,+\infty) \mapsto \ee(v;\LN,\Lz)$ is smooth and the following hold:
		\begin{enumerate}[i)]
			
			\item there exists a sequence $\{f_{j}\}_{j \in \Nz} \subset \R$ such that
				\begin{equation}\label{eq:asyev0}
				 	\ee(v;\LN,\Lz) \sim \sum_{j = 0}^{+\infty} f_j \,v^{2j}\,,
				 	\qquad \mbox{as $v \to 0^+$}\,.
				\end{equation}							
			
			\item there exist $\{g_j\}_{j \in \Nz} \subset \R$, $\{K_j\}_{j \in \Nz}  \subset \Nz$, $\{h_{j,k}\}_{j,k \in \Nz} \subset \C$ and $\{d_{j,k}\}_{j,k \in\Nz} \subset \R_+$ such that
				\begin{equation}\label{eq:asyeving}
				 	\ee(v;\LN,\Lz) \sim \sum_{j = 0}^{+\infty} v^{-j-2} \bigg[ g_j + \Re\bigg(\sum_{k = 0}^{K_j} h_{j,k}\,e^{i d_{j,k} v}\bigg) \bigg] \,,
				 	\qquad \mbox{as $v \to +\infty$}\,.
				\end{equation}						
														
		\end{enumerate}
	\end{proposition}

	\begin{remark}\label{rem:gj}
		All coefficients $f_j$ in \eqref{eq:asyev0} and $g_{j},K_j,h_{j,k},d_{j,k}$ in \eqref{eq:asyeving} can be computed explicitly.
		Especially, we record that
		\begin{equation}
			f_0 = \frac{1}{4 \pi^2} \!\sum_{\mm,\nn=1}^N \!\big[\GaN^0\big]_{\mm\nn}^{-1}\,,	\label{eq:f0}
		\end{equation}		
		while,
		\begin{equation}\label{eq:gj}
			g_j = \begin{cases}
			 \frac{(-1)^{j/2}}{\pi} \tx \sum_{\nn=1}^N (4\pi \alpha_\nn)^{j+1}\,, & \mbox{for even $j$}\,, \vspace{0.1cm}\\
			 0\,, & \mbox{for odd $j$}\,,
			\end{cases}
		\end{equation}		
		and
		\begin{equation}\label{eq:eevinf}
			\ee(v;\LN,\Lz) 
			= \frac{4}{v^2} \Bigg(\sum_{\nn=1}^N  \alpha_\nn - \frac{1}{4\pi}\sum_{\mm \neq \nn} \frac{\cos(2 v |x_\mm-x_\nn|)}{|x_\mm-x_\nn|}\Bigg) + O(v^{-3})\,, \qquad \mbox{as $v \to +\infty$}\,.	
		\end{equation}
		Here and in \eqref{eq:asyeving} the non-oscillating and oscillating terms have been distinguished as they are going to play different roles in the subsequent analysis.

		Two further features emerge from the proof of \cref{prop:ee}. First, the coefficients $d_{j,k}$ are given by linear combinations of the mutual distances between the points. Second, the leading-order contribution in the expansion \eqref{eq:asyeving} for $v \to +\infty$ would be in principle of order $v^{-1}$, but the corresponding coefficient vanishes due to an exact cancellation. Other coefficients may vanish as well, but this has not been investigated further as it does not affect the following discussion.
	\end{remark}

Building on \cref{prop:ee} and using the representation \eqref{eq:zetaev}, we obtain the following.

	\begin{theorem}\label{thm:zLNLz}
		Assume \eqref{eq:ass1} and \eqref{eq:ass2}. Then, the map $s \in \C \mapsto \zeta(s;\LN,\Lz)$ is meromorphic, with possible simple poles at $s \in 1/2 + \Z$.
	\end{theorem}

	\begin{remark}
		The residues of $\zeta(s;\LN,\Lz)$ at the negative values $s \in -1/2 -\Nz$ are entirely determined by the coefficients $g_j$ in \eqref{eq:gj}. In particular, they depend on the interaction parameters $\alpha_\nn$, but not on the positions $x_\nn$ of the obstacles.	
	\end{remark}
	
\subsection{The renormalized Casimir energy and thermal contributions}\label{sec:Zren}		
We now turn to the physical application and compute the renormalized connected partition function for the model under analysis, defined according to \eqref{eq:lnZN}.

	\begin{theorem}\label{thm:ZZren}
		Assume \eqref{eq:ass1} and \eqref{eq:ass2}. Let $\beta > 0$ be the inverse temperature and $\ell > 0$ the length scale parameter. Then, the renormalized connected relative partition function is given by	
		\begin{equation}\label{eq:ZZren}
			\big[\ln \ZeN - \ln \Zez\big]_{\ren} = - \mathcal{E}_{\vac}\,\beta - \log\eta(\beta;\LN,\Lz)\,,
		\end{equation}
		where, for $\vz > 0$ fixed arbitrarily,
		\begin{multline}\label{eq:Ecasdef}
			\mathcal{E}_{\vac} := 2 \big[1-\log(2\ell \vz)\big] \sum_{n =1}^N \alpha_n 
			+ \frac{1}{2} \int_0^{\vz}\! dv\;v\; \ee(v;\LN,\Lz)  \\
			+ \frac{1}{4\pi} \sum_{\mm \neq \nn} \left[ \frac{\sin(2 \vz|x_\mm - x_\nn|\,)}{\vz\,|x_\mm - x_\nn|^2} 
				-  \int_{\vz}^{\infty}\! dv\; \frac{\sin(2v |x_\mm - x_\nn|\,)}{v^2\,|x_\mm-x_\nn|}\right]
\\					 
		+ \frac{1}{2} \int_{\vz}^{\infty}\! dv\;v \bigg[\ee(v;\LN,\Lz) - \frac{4}{v^2} \bigg(\sum_{\nn=1}^N  \alpha_\nn - \sum_{\mm \neq \nn} \frac{\cos(2 v\,|x_\mm-x_\nn|)}{4\pi|x_\mm-x_\nn|}\bigg) \bigg]\,,
		\end{multline}
		is the {\sl Casimir energy} (or {\sl vacuum energy}) and
		\begin{equation}\label{eq:logndef}
			\eta(\beta;\LN,\Lz) := \exp \left[\int_0^\infty dv\, \log(1-e^{-\beta v})\,\ee(v;\LN,\Lz) \right] .
		\end{equation}
		is the {\sl relative Dedekind eta function}.
	\end{theorem}

	\begin{remark}
		All the integrals appearing in \eqref{eq:Ecasdef} and \eqref{eq:logndef} are absolutely convergent, owing to the results established in \cref{prop:ee}. Furthermore, the expression on the right-hand side of \eqref{eq:Ecasdef} is independent of $v_0 > 0$. The proof also makes clear that only the Casimir contribution \eqref{eq:Ecasdef} requires renormalization via analytic continuation, whereas $\log\eta(\beta;\LN,\Lz)$ is well defined from the outset.
	\end{remark}
	
	\begin{remark}\label{rem:ellind}
		The renormalization length $\ell > 0$ appears only in the first term on the right-hand side of \eqref{eq:Ecasdef}, defining $\mathcal{E}_{\vac}$. Crucially, this contribution depends solely on the interaction strengths $\alpha_\nn$, while being independent of the positions of the obstacles. As a consequence, the forces generated by vacuum fluctuations do not depend on $\ell$, since they are determined by derivatives of $\mathcal{E}_{\vac}$ with respect to the points positions. The same argument holds for thermal contributions. The parameter $\ell$ therefore serves as a purely auxiliary coefficient, and does not affect the physically measurable quantities.
	\end{remark}

Despite the rather cumbersome expressions appearing in \eqref{eq:Ecasdef}, \cref{thm:ZZren} provides a fully explicit expression for the renormalized partition function that can be used to derive the main thermodynamic quantities. In particular, one can consider the renormalized {\sl Helmholtz free energy}, {\sl internal energy} and {\sl entropy} at inverse temperature $\beta > 0$, respectively defined as follows:
\begin{align}
	\Fr(\beta) &:= - \tfrac{1}{\beta} \big[\ln \ZeN - \ln \Zez\big]_{\ren}\,; \label{eq:Frdef} \\
	\Ur(\beta) &:= - \partial_\beta \big[\ln \ZeN - \ln \Zez\big]_{\ren}\,; \label{eq:Urdef}\\
	\Sr(\beta) &:= \beta\,\Ur + \big[\ln \ZeN - \ln \Zez\big]_{\ren}\,. \label{eq:Srdef}
\end{align}

	\begin{corollary}\label{cor:FUS}
		Assume \eqref{eq:ass1} and \eqref{eq:ass2}. Then, the following asymptotic expansions hold.
		\begin{enumerate}[i)]
		\item {\sl Low-temperature limit}, for $\beta \to +\infty$:
			\begin{align}
				\Fr(\beta) & \sim \mathcal{E}_{\vac} - \frac{\pi^2}{\beta^2} \sum_{j = 0}^{+\infty} \frac{(2\pi)^{2j}\,|B_{2j+2}|\,f_j}{(j+1)(2j+1)}\,\frac{1}{\beta^{2j}} \,, \vspace{0.1cm}\label{eq:Frbinf}\\
				\Ur(\beta) & \sim \mathcal{E}_{\vac} + \frac{\pi^2}{\beta^2} \sum_{j = 0}^{+\infty} \frac{(2\pi)^{2j}\,|B_{2j+2}|\,f_j}{(j+1)}\,\frac{1}{\beta^{2j}}\,, \vspace{0.1cm} \label{eq:Urbinf}\\
				\Sr(\beta) & \sim \frac{2\pi^2}{\beta} \sum_{j = 0}^{+\infty} \frac{(2\pi)^{2j}\,(j+1)\,|B_{2j+2}|\,f_j}{(j+1)(2j+1)}\,\frac{1}{\beta^{2j}}\,,\label{eq:Srbinf}
			\end{align}
		where $\{B_{2j+2}\}_{j \in \Nz}$ are the even Bernoulli numbers and $\{f_{j}\}_{j \in \Nz} \subset \R$ are the coefficients in \eqref{eq:asyev0}.
		
		\item {\sl High-temperature limit}, for $\beta \to 0^+$:
			\begin{align}
				\Fr(\beta) &= \left(\int_0^\infty\!\! dv\,\ee(v;\LN,\Lz)\right) \frac{\log\beta}{\beta}
							 + \left(\int_0^\infty\!\! dv\,\log v\;\ee(v;\LN,\Lz)\right) \frac{1}{\beta} + \mathcal{O}(\log\beta)\,, \vspace{0.1cm}\label{eq:Frb0} \\
				\Ur(\beta) &= \left(\int_0^\infty\!\! dv\; \ee(v;\LN,\Lz)\right) \frac{1}{\beta}  + \mathcal{O}(\log\beta) \,, \vspace{0.1cm} \label{eq:Urb0}\\
				\Sr(\beta) &= 
				- \left(\int_0^\infty\!\! dv\,\ee(v;\LN,\Lz)\right) \log\beta
				+ \int_0^\infty\!\! dv\; (1-\log v)\, \ee(v;\LN,\Lz) 
				+ \mathcal{O}(\beta \log\beta)\,. \label{eq:Srb0}
			\end{align}
			Furthermore all the integrals in the above equations are absolutely convergent.
		\end{enumerate}
	\end{corollary}
	
	\begin{remark}
		The asymptotic relations \eqref{eq:Frbinf}-\eqref{eq:Srbinf} imply, in particular,
		\begin{equation}
				\lim_{\beta \to +\infty} \Fr(\beta) = \lim_{\beta \to +\infty} \Ur(\beta) = \mathcal{E}_{\vac}\,, \qquad
				\lim_{\beta \to +\infty} \Sr(\beta) = 0\,.\label{eq:FSrlim}
			\end{equation}
		This justifies, {\em a posteriori}, the designation ``Casimir energy'' for the constant $\mathcal{E}_{\vac}$: it is the contribution to both the free and internal energies arising purely from vacuum fluctuations of the quantum field, in the absence of thermal excitations.
	\end{remark}
	
	\begin{remark}
		For the reader's convenience, we recall that the Bernoulli numbers appearing in \eqref{eq:Frbinf}-\eqref{eq:Srbinf} are given by \cite[Eq. 24.6.1]{OL+10}
			\begin{equation*}
				B_{2j+2} = \sum_{k=2}^{2j+3} \frac{(-1)^{k-1}}{k} \binom{2j+3}{k} \sum_{h=1}^{k-1} h^{2j+2}\,, \qquad \mbox{for any $j \in \Nz$}\,.
			\end{equation*}
		In particular, for $j = 0$, one has $B_2 = 1/6$. Combining this with the expression \eqref{eq:f0} for $f_0$, the leading thermal corrections in the low-temperature limit $\beta \to +\infty$ read as follows:
			\begin{gather*}
				\Fr(\beta) = \mathcal{E}_{\vac} - \bigg(\frac{1}{24} \!\sum_{\mm,\nn=1}^N \big[\GaN^0\big]^{-1}_{\mm\nn}\bigg)  \frac{1}{\beta^2} + \mathcal{O}(\beta^{-4}) \,; \vspace{0.1cm}\\
				\Ur(\beta) = \mathcal{E}_{\vac} + \bigg(\frac{1}{24} \!\sum_{\mm,\nn=1}^N \big[\GaN^0\big]^{-1}_{\mm\nn}\bigg)  \frac{1}{\beta^2} + \mathcal{O}(\beta^{-4})\,; \vspace{0.1cm} \\
				\Sr(\beta) = \bigg(\frac{1}{12} \!\sum_{\mm,\nn=1}^N \big[\GaN^0\big]^{-1}_{\mm\nn}\bigg)  \frac{1}{\beta} + \mathcal{O}(\beta^{-3})\,. 
			\end{gather*}
		Here $\GaN^0$ is just the $N \times N$ matrix obtained by setting $z = 0$ in the definition \eqref{eq:GaN} of $\GaN^z$.
		The growth of the Bernoulli numbers $B_{2j+2} \sim (2j)^{2j+5/2}$ as $j \to +\infty$ \cite[Eq.~24.11.1]{OL+10} suggests that the series in \eqref{eq:Frbinf}--\eqref{eq:Srbinf} are most likely divergent for any $\beta > 0$. A rigorous assessment of this claim, however, would require a detailed analysis of the asymptotic behavior of the coefficients $f_j$, which lies beyond the scope of this work. We limit ourselves to regard the said series as asymptotic expansions, possibly with zero radius of convergence.
	\end{remark}

	\begin{remark}
		In the low-temperature regime $\beta \to +\infty$, the asymptotics are entirely governed by infrared contributions. In contrast, in the high-temperature limit $\beta \to 0^+$, the infrared and ultraviolet sectors contribute equally, as reflected by the full-range integrals of $\ee(v;\LN,\Lz)$ in \eqref{eq:Frb0}-\eqref{eq:Srb0}. This is what makes the computation of higher-order terms in this regime considerably more involved.
	\end{remark}
	
	\begin{remark}	
		The form of the low- and high-temperature expansions derived in \cref{cor:FUS} is standard in Casimir-type problems involving a scalar field subject to external perturbations. They can in general be related, respectively, to the large- and small-time asymptotics of the relative heat trace $\Tr[e^{-t\LN} - e^{-t \Lz}]$, see {\em e.g.} \cite{DK78} and \cite[\S 5.2]{BK+09}.
	\end{remark}

\subsection{Born series expansion for the Casimir energy}\label{sec:Born}
Having established how to compute the renormalised Casimir energy, we now derive an alternative representation that lends itself to a clearer physical interpretation and further proves to be more convenient for numerical implementations.

	\begin{theorem}\label{thm:EvacBorn}
		Assume \eqref{eq:ass1} and \eqref{eq:ass2}. Then, the vacuum energy can be expressed as 
			\begin{equation}\label{eq:Eren2}
				\mathcal{E}_{\vac} = \mathcal{E}_{0}^{\ren} + \mathcal{E}_{1}^{\ren} + \sum_{j = 2}^{+\infty} \mathcal{E}_{j}\,,
			\end{equation}
		where
			\begin{equation}\label{eq:E0ren}
				\mathcal{E}_{0}^{\ren} := 2 \sum_{\nn=1}^N \alpha_\nn\, \big[1-\log(8\pi \alpha_\nn \ell) \big] \,,
			\end{equation}
			\begin{equation}\label{eq:E1ren}
				\mathcal{E}_{1}^{\ren} :=  - \frac{1}{4\pi} \sum_{\mm \neq \nn} \frac{1}{|x_\mm-x_\nn|^2}\int_0^{\infty}\! dv\; \Im \bigg( \frac{(4\pi \alpha_\mm)(4\pi \alpha_\nn) + v^2}{(4\pi \alpha_\mm - i v)^2(4\pi \alpha_\nn - i v)^2} \;e^{2i v\,|x_\mm-x_\nn|}\bigg) \,,
			\end{equation}
			\begin{equation}\label{eq:Ejren}
				\mathcal{E}_{j} := \frac{1}{8\pi^2}\int_0^{\infty} dv\;v\; \Re \left(\sum_{\mm,\nn=1}^N \Big[ \Big(\big(\VNp \big)^{-1}\PNp\Big)^{j} (\VNp)^{-1} \Big]_{\mm\nn} \,e^{i v |x_\mm-x_\nn|}\right), \quad\; \mbox{for $j \geqslant 2$}\,,
			\end{equation}
			and $\VNp,\PNp$ are the $N \times N$ complex-valued matrices with components
			\begin{equation}\label{eq:VpPp}
				\big[\VNp\big]_{\mm\nn} := \left(\alpha_\mm - \frac{i v}{4\pi}\right)\delta_{\mm\nn}\,, \qquad
				\big[\PNp\big]_{\mm\nn} := \frac{e^{i v\,|x_\mm-x_\nn|}}{4\pi|x_\mm-x_\nn|}\,(1-\delta_{\mm\nn}) \,.
			\end{equation}
			
			Moreover, the series in \eqref{eq:Eren2} is absolutely convergent and, for any $J \geqslant 2$,
			\begin{equation}\label{eq:ErenRest}
				\sum_{j=J}^{+\infty} \big|\mathcal{E}_{j}\big|
				\leqslant 
				\frac{4N\,(\max_\mm \alpha_\mm)^2}{\min_\mm \alpha_\mm} \,\rho^{J-2}						
				 \min\left\{\frac{\rho}{(J-1)(1-\rho)}\;,\, \big|\log(1-\rho)\big| \right\},
			\end{equation}
			where
			\begin{equation}\label{eq:rho}
				\rho := \bigg(\sum_{\mm \neq \nn} \frac{1}{(4\pi \alpha_\mm)^2\, |x_\mm - x_\nn|^2}\bigg)^{\!1/2} \in\; (0,1)\,.
			\end{equation}
	\end{theorem}

Retracing the steps leading to \cref{thm:EvacBorn}, one can see that, for any $j \in \Nz$, the energy contribution $\mathcal{E}_{j}^{(\ren)}$ in \eqref{eq:Eren2} is a (possibly renormalized) version of the formal expression
		\begin{multline}\label{eq:Ejform}
			\mathcal{E}_{j} = \lim_{\varepsilon \to 0^+} \int_0^{\infty} \!\!\!dv\;v \int_{\R^3}\! dx\! \sum_{\mm,\nn=1}^N\! \frac{1}{2i\pi}\!\left(\sqrt{v^2 e^{i \varepsilon}}\, (P_N^{v^2 e^{i \varepsilon}})_{x \mm} \Big[ \Big(\big(V_N^{v^2 e^{i\varepsilon}} \big)^{-1} P_N^{v^2 e^{i \varepsilon}}\Big)^{j} (V_N^{v^2 e^{i\varepsilon}})^{-1} \Big]_{\mm\nn} (P_N^{v^2 e^{i \varepsilon}})_{\nn x} \right.\\
			\left. - \sqrt{v^2 e^{i (2\pi - \varepsilon)}}\, (P_N^{v^2 e^{i (2\pi - \varepsilon)}})_{x \mm} \Big[ \Big(\big(V_N^{v^2 e^{i(2\pi - \varepsilon)}} \big)^{-1} P_N^{v^2 e^{i (2\pi - \varepsilon)}}\Big)^{j} (V_N^{v^2 e^{i(2\pi - \varepsilon)}})^{-1} \Big]_{\mm\nn} (P_N^{v^2 e^{i (2\pi - \varepsilon)}})_{\nn x}\right),
		\end{multline}
where, in agreement with \eqref{eq:VpPp},
		\begin{equation*}
			\big[\VN\big]_{\mm\nn} = \left(\alpha_\nn - \frac{i\sqrt{z}}{4\pi}\right)\delta_{\mm\nn} \,, \qquad
			\big[\PN\big]_{\mm\nn} = \Gz^{z}(x_\mm-x_\nn)\,(1-\delta_{\mm\nn})\,.
		\end{equation*}
and, with an obvious abuse of notation, we have also set
	\begin{align*}
		(P_N^{z})_{x \mm} = \Gz^{z}(x - x_\mm)\,, \qquad
		(P_N^{z})_{\nn x} = \Gz^{z}(x_\nn - x) \,.
	\end{align*}	
	
Switching to a notation closer to that commonly used in the physics literature, \eqref{eq:Ejform} can be recast in a more compact form as
		\begin{multline}\label{eq:bornterm}
			\mathcal{E}_{j} = \int_0^{\infty}\!\!\! dv\;v\! \int_{\R^3}\!\!\! dx \;\tfrac{1}{2\pi i}\left[\sqrt{v^2 + i 0}\,P_N^{v^2 +  i 0} \Big(\big(V_N^{v^2 + i 0} \big)^{-1} P_N^{v^2 + i 0}\Big)^{j+1} \right.\\
			\left. - \sqrt{v^2 - i 0} \,P_N^{v^2 -i 0} \Big(\big(V_N^{v^2 -i 0} \big)^{-1} P_N^{v^2 - i 0}\Big)^{j+1} \right]_{xx}.
		\end{multline}
This identity naturally suggests to interpret \eqref{eq:Eren2} as a \emph{Born series} for the renormalized Casimir energy, in which the summation index $j\in\Nz$ counts the number of interactions among the point-like obstacles, mediated by the scalar field. A diagrammatic representation of this picture is provided in \cref{fig:born}. 		

	\begin{figure}[t!]
		\includegraphics[width=17cm]{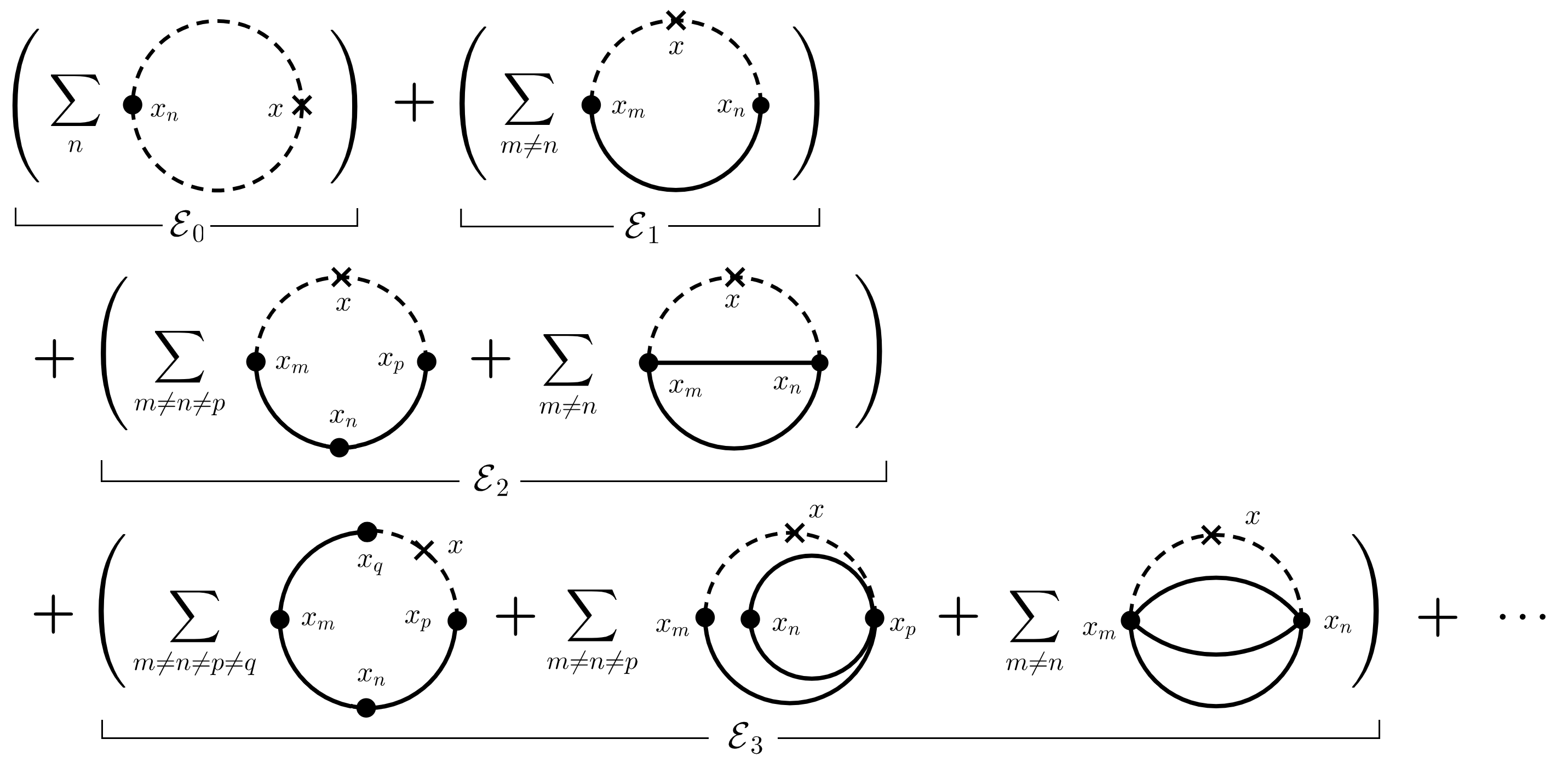}
 		\caption{
 			Graphical representation of the Born series for $\mathcal{E}_{\vac}$, see \eqref{eq:Eren2}.
		}
	    \label{fig:born}
	\end{figure}

Within this picture, the diagonal matrices $(V_{N}^{v^2\pm i0})^{-1}_{\mm\nn}$ correspond to (dressed) interaction vertexes (depicted as bullet points), while the matrices $(P_{N}^{v^2\pm i0})_{\mm\nn}$ describe propagation between distinct obstacles (shown as solid lines). Since $(P_{N}^{v^2\pm i0})_{\mm\nn}$ is strictly off-diagonal, each factor necessarily connects two \emph{different} points, thereby excluding self-interactions at the level of individual propagators (no single-line loops). Nonetheless, compositions of such propagators can produce closed paths connecting two or more vertexes. A simple combinatorial argument further shows that each term $\mathcal{E}_j$ receives contributions from diagrams with a number of vertexes ranging from $2$ to $j+1$.	
The vectors $(P_{N}^{v^2\pm i0})_{x\mm}$ and $(P_{N}^{v^2\pm i0})_{\nn x}$ (represented by dashed lines) connect a given physical obstacle to an arbitrary point in space $x\in\mathbb{R}^3$, where no interaction occurs. The integration over $x$ formally sums over all possible positions of this point, resulting in an expected volume factor. Finally, the difference of the terms in the two lines of \eqref{eq:bornterm} can be interpreted as the subtraction of advanced and retarded contributions, thereby selecting the causal part of the expansion.

The above discussion makes evident that Casimir interactions in the model under consideration can indeed be viewed as arising from multiple-scattering processes, consistently with other scattering-theoretic formulations of Casimir physics available in the literature \cite{RE+09,RCD+10,B14,EB23}.
	\medskip
	
We conclude this paragraph with some additional comments.	
	
	\begin{remark}
		Only the lowest-order contributions, namely the single-centers term ($j=0$) and the pairwise interaction term ($j=1$), require renormalization. All higher-order terms, corresponding to processes involving at least two successive interactions, are automatically well defined and finite, reflecting the improved ultraviolet behavior induced by multiple scattering processes. 
	\end{remark}	
	
	\begin{remark}
		The proof of \cref{thm:EvacBorn} is based on a corresponding Born-type expansion for the relative spectral density $\ee(v;\LN,\Lz)$, see \cref{sec:bornexp}. While this representation is well adapted to the analysis of the vacuum energy, it does not seem to be equally effective in the study of thermal corrections, which motivates our present restriction to the zero-temperature contribution.
	\end{remark}
	
	\begin{remark}
		Concerning the rate of convergence of the series in \eqref{eq:Eren2}, the remainder estimate \eqref{eq:ErenRest} provides efficient control when $J$ is large and $\rho$ is sufficiently small. This bound deteriorates as $\rho \to 1^{-}$, that is, in the regime where the assumption \eqref{eq:ass2} is saturated. Yet, it is worth noting that each term $\mathcal{E}_{j}$ as in \eqref{eq:Ejren} remains well defined even if \eqref{eq:ass2} is violated. In this regime, however, the derivation provided in this work is no longer justified, and the series in \eqref{eq:Eren2} may even fail to converge.
	\end{remark}
		
	\begin{remark}\label{rem:E0}
		As previously observed in \cref{rem:ellind}, the forces exerted by the quantum field on the point obstacles are independent of the renormalization length $\ell$. This is consistent with the Born expansion of \cref{thm:EvacBorn}, since $\ell$ enters only through the single-vertex term $\mathcal{E}_0$, which is itself independent of the obstacles' positions.
		To say more, \cref{thm:EvacBorn} shows that, apart from the constant $\mathcal{E}_0$, the renormalised Casimir energy depends on the obstacles' configuration solely through their mutual distances:
		\begin{equation}
   			 \mathcal{E}_{\mathrm{vac}} = \mathcal{E}_0^{\ren} + \mathcal{E}^{(\mathrm{int})}_{\mathrm{vac}} \big(\{\,|x_{\mm}-x_{\nn}|\,\}_{\mm \neq \nn}\big)\,.
		\end{equation}
		As a consequence, the force on the $\nn$-th obstacle, computed taking the gradient with respect to $x_\nn \in \R^3$,
		\begin{equation}
			\bm{F}_\nn = - \nabla_{\!\nn}\, \mathcal{E}_{\vac}\,,
		\end{equation}
		takes the form
		\begin{equation}
   		 	\bm{F}_{\nn}(x_1,\dots,x_N) = \sum_{\mm \neq \nn} F_{\mm\nn}(x_1,\dots,x_N)\, \frac{x_{\mm}-x_{\nn}}{|x_{\mm}-x_{\nn}|}\,,
		\end{equation}
		where the sum runs only over the index $\mm \in \{1,\dots,N\} \setminus \{\nn\}$, for fixed $\nn$, and
		\begin{equation}
			F_{\mm\nn} := \frac{\partial \mathcal{E}^{(\mathrm{int})}_{\mathrm{vac}}}{\partial|x_\mm-x_\nn|}\,.
		\end{equation}				
		This representation reveals two key features of the interaction: the force on $x_{\nn}$ decomposes into pairwise contributions, each directed along the line joining $x_{\nn}$ to $x_{\mm}$; yet, the scalar intensities $F_{\mm\nn}$ depend on the full configuration $(x_1,\dots,x_N)$, and not just on the pairwise distance $|x_{\mm} - x_{\nn}|$. The interaction is therefore genuinely many-body in nature: although it decomposes into pairwise directional contributions, the strength of each is modulated by the presence and positions of all other obstacles.
	\end{remark}

\subsection{The case of identical obstacles and some numerical examples}\label{sec:num}
In this final paragraph we present a collection of numerical experiments that illustrate and complement the analytical results established in the previous sections. 
Although the Born series of \cref{thm:EvacBorn} is, in principle, already amenable to numerical evaluation, its direct implementation can become computationally demanding.
To simplify the analysis and make the computations more tractable with the resources at our disposal, we now restrict our attention to the case of identical obstacles. Accordingly, in addition to the standing hypotheses \eqref{eq:ass1} and \eqref{eq:ass2}, we impose
	\begin{equation}\label{eq:aannaa}
		\alpha_\nn = \alpha \,, \qquad \mbox{for all $\nn \in {1,\dots,N}$}\,.
	\end{equation}
This condition introduces an additional symmetry that allows for further explicit simplifications, significantly reducing the computational complexity.

In this setting, the single obstacle scattering length $(4\pi\alpha)^{-1}$ provides a natural length unit. We rescale the obstacles' positions $x_\nn \in \R^3$ accordingly, introducing the dimensionless coordinates
	\begin{equation}\label{eq:yn}
		y_\nn := 4\pi \alpha\,x_\nn \,\in\,\R^3 .
	\end{equation}
Recalling that we are working in natural units with $\hbar = c = 1$, the above position \eqref{eq:aannaa} further identifies $4\pi \alpha$ as a natural energy scale. We shall henceforth refer to the rescaled Casimir energy
	\begin{equation}\label{eq:Etilde}
		\widetilde{\mathcal{E}}_{\vac} := \frac{1}{4\pi \alpha}\, \mathcal{E}_{\vac}\,. 
	\end{equation}

	\begin{corollary}\label{cor:num}
		Assume \eqref{eq:ass1}, \eqref{eq:ass2} and \eqref{eq:aannaa}. Then, the rescaled Casimir energy can be expressed as
		\begin{equation}\label{eq:Erennum}
			\widetilde{\mathcal{E}}_{\vac} = \widetilde{\mathcal{E}}_{0}^{\ren} + \widetilde{\mathcal{E}}_{1}^{\ren} + \sum_{j = 2}^{+\infty} \widetilde{\mathcal{E}}_{j}\,,
		\end{equation}
		where
		\begin{equation}\label{eq:E0numc}
			\widetilde{\mathcal{E}}_{0}^{\ren} = \frac{N}{2\pi}\, \big[1-\log(8\pi \alpha \ell) \big] \,,
		\end{equation}
		\begin{equation}\label{eq:E1numc}
			\widetilde{\mathcal{E}}_{1}^{\ren} 
				= \frac{1}{\pi} \sum_{\mm \neq \nn} \left[ \frac{1}{r^2}\,\Big( \KK_{2}(r) - 2\,\KK_{3}(r) \Big) \right]_{r \,=\, 2|y_\mm-y_\nn|} ,
		\end{equation}
		and, for all integer $j \geqslant 2$,
		\begin{multline}\label{eq:Ejnumc}
			\widetilde{\mathcal{E}}_{j} = \frac{1}{2\pi}  \sum_{\mm \neq p_1 \neq \dots \neq p_{j-1} \neq \nn} \frac{1}{|y_\mm - y_{p_1}| \cdot |y_{p_1} - y_{p_2}| \cdot \dots \cdot |y_{p_{j-1}} - y_{\nn}|} \;\times \\
				\times\, \Big[ \KK_{j+1}(r) - \KK_j(r) \Big]_{r \,=\, |y_\mm - y_{p_1}| + |y_{p_1} - y_{p_2}| +\, \dots + |y_{p_{j-1}} - y_{\nn}| + |y_{\nn}-y_{\mm}|}\,.
		\end{multline}
		The sum in \eqref{eq:Ejnumc} runs over all ordered collections $(\mm,p_1,\dots,p_{j-1},\nn) \in \{1,\dots,N\}^{\times (j+1)}$ such that no two consecutive indexes coincide.
		For $r > 0$ and integers $k \geqslant 2$, the functions $\KK_k(r)$ are defined by
		\begin{equation}\label{eq:KKk}
		\KK_k(r) := \frac{1}{(k-1)!} \left[\, \sum_{h = 0}^{k-2} (k-h-2)!\,(-r)^h + (- r)^{k-1}\,e^{r} E_1(r)\,\right] ,
		\end{equation}
		with $E_1(\,\cdot\,)$ denoting the exponential integral function \cite[\S 6.2(i)]{OL+10}.	
		
		Moreover, for any integer $J \geqslant 2$,
		\begin{equation}\label{eq:ErenRestnum}
				\sum_{j=J}^{+\infty} \big|\widetilde{\mathcal{E}}_{j}\big|
				\leqslant \frac{N}{\pi}\,\bigg[\rho^{J-2}\, \min\left\{\frac{\rho}{(J-1)(1-\rho)}\,,\,\big|\log(1-\rho)\big|\right\}\Bigg]_{\rho \,=\, \left(\sum_{\mm \neq \nn} \frac{1}{|y_\mm - y_\nn|^2}\right)^{\!1/2}} \,.
		\end{equation}
	\end{corollary}
\smallskip

We now examine some case studies with an assigned number of obstacles. For any $J \geqslant 2$, we refer to the approximate, rescaled vacuum interaction energy
	\begin{equation}\label{eq:Eintnum}
		\widetilde{\mathcal{E}}^{(\mathrm{int})}_{J} := \widetilde{\mathcal{E}}_{1}^{\ren} + \sum_{j = 2}^{J} \widetilde{\mathcal{E}}_{j}\,,
	\end{equation}
and to the associated relative error
	\begin{equation}
		\widetilde{\mathcal{R}}_{J} := \left(\tx\sum_{j = J + 1}^{+\infty} \big|\widetilde{\mathcal{E}}_{j}\big|\right) \Big/ \left| \widetilde{\mathcal{E}}^{(\mathrm{int})}_{J} \right|.
	\end{equation}
In what follows, for given $N \geqslant 1$, we evaluate $\widetilde{\mathcal{E}}^{(\mathrm{int})}_{J}$ by fixing $N-1$ obstacles and letting the position of the $N$-th obstacle vary in $\R^3$, respecting the condition \eqref{eq:ass2}.
\medskip

\paragraph{{\bf N = 1.}}
In the case of a single point-like obstacle, only the self-interaction term contributes, and the rescaled Casimir energy reduces to
	\begin{equation}\label{eq:E1p}
		\widetilde{\mathcal{E}}_{\vac} 
		= \widetilde{\mathcal{E}}_{0}^{\ren}
		= \frac{1}{2\pi}\, \big[1 - \log(8\pi\alpha  \ell )\big] \, .
	\end{equation}
This is consistent with the expression for the renormalized partition function reported in \cite[\S 4.1]{SZ09}.\footnote{The first equation on p.~174 of the cited reference contains a misprint. In fact, combining the preceding formulas therein, one obtains (in the notation of that paper) $$ \log Z = 2\big(\log(8\pi \alpha \ell) - 1\big)\alpha \beta - \log\eta(\beta;-\Delta_\alpha,-\Delta)\,. $$}
\medskip

\paragraph{{\bf N = 2.}}
For two obstacles, only the relative distance between the two obstacles matters, namely
	\begin{equation}
		d_{12} := |y_1 - y_2| \in (\sqrt{2},+\infty)\,.
	\end{equation}
Here, the admissible values of $d_{12}$ are determined by the hypothesis \eqref{eq:ass2}. 
	
	\begin{figure}[t!]
		\centering
		\begin{subfigure}{0.48\textwidth}
    			\centering
    			\includegraphics[width=\linewidth]{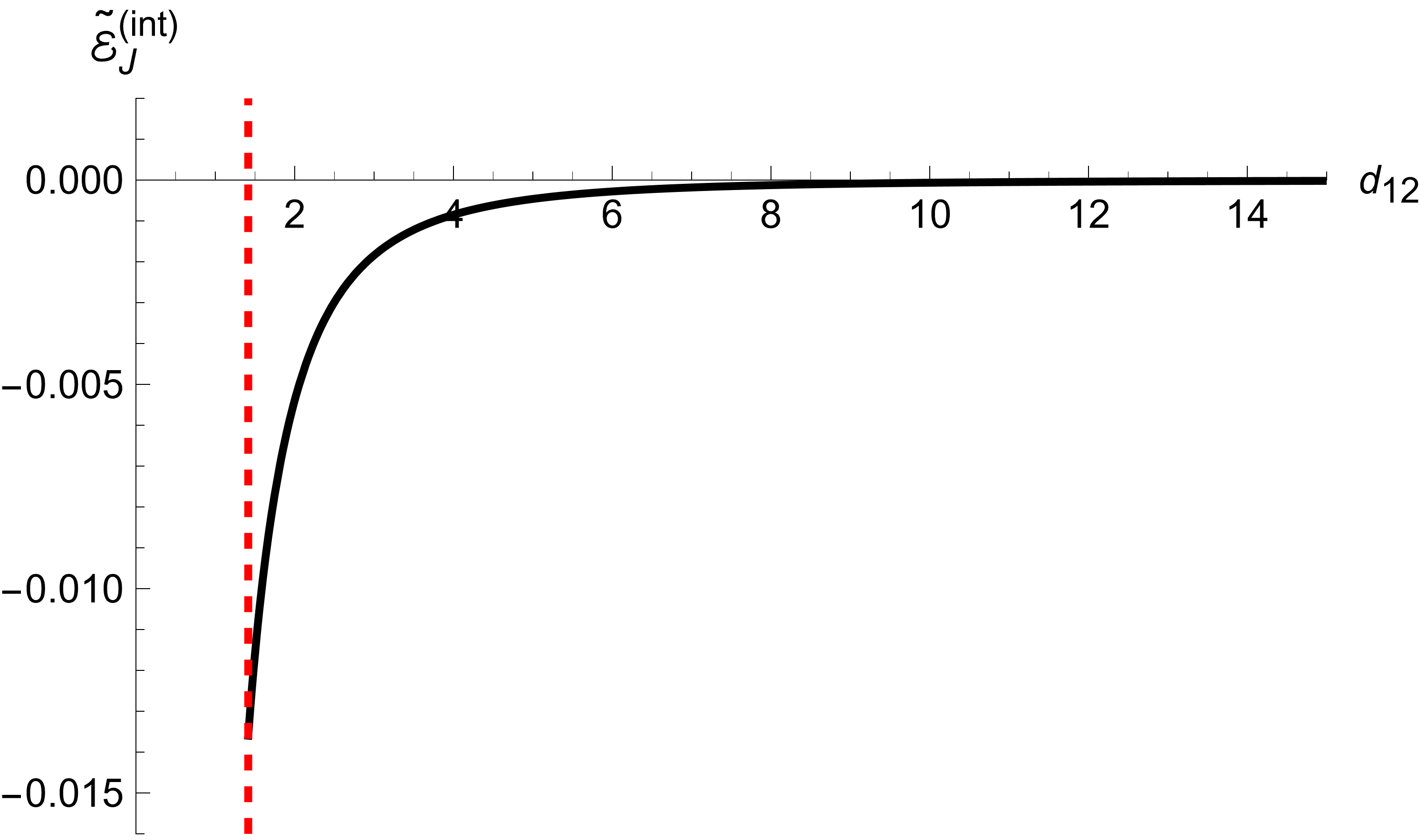}\label{fig:2a}
   			\caption{}
		\end{subfigure}
			\hfill
		\begin{subfigure}{0.48\textwidth}
   			\centering
   			\includegraphics[width=\linewidth]{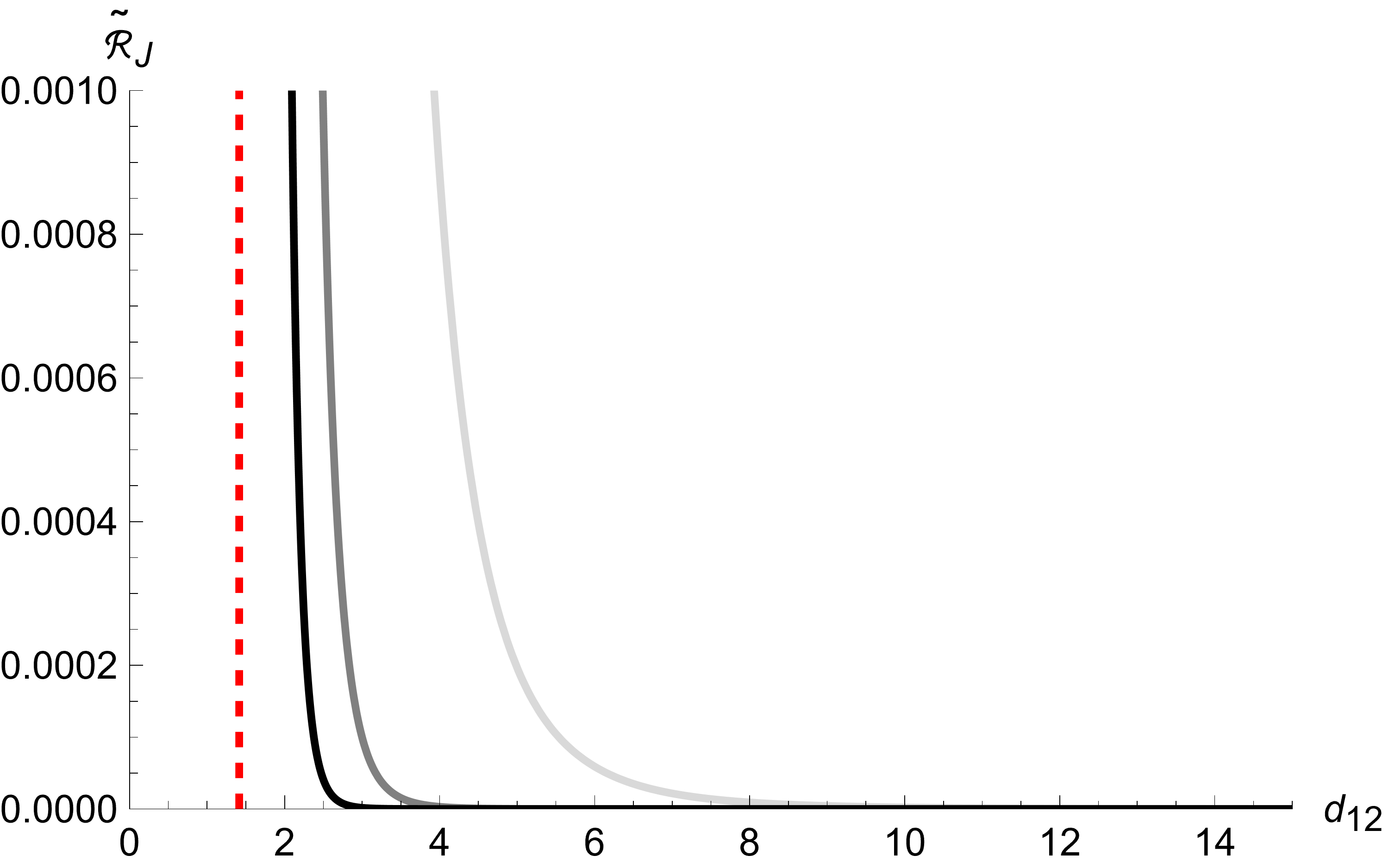}\label{fig:2b}
   			\caption{}
		\end{subfigure}
		\caption{(A) $\widetilde{\mathcal{E}}^{(\mathrm{int})}_{J}$ for $N = 2$ and $J = 15$, plotted as a function of $d_{12}$. (B) $\widetilde{\mathcal{R}}_{J}$ for $N = 2$ and $J = 10$ (light gray), $J = 15$ (gray), $J = 20$ (black), plotted as a function of $d_{12}$. In both panels, the red dashed line marks the cutoff $d_{12} > \sqrt{2}$ imposed by \eqref{eq:ass2}.}
	\label{fig:E12}
	\end{figure}

\Cref{fig:E12} displays the approximate rescaled vacuum interaction energy $\widetilde{\mathcal{E}}^{(\mathrm{int})}_{J}$ for $J = 15$, together with the relative error $\widetilde{\mathcal{R}}_{J}$ for selected values of $J$. The profile of $\widetilde{\mathcal{E}}^{(\mathrm{int})}_{J}$ remains essentially unchanged throughout the range $3 \leqslant J \leqslant 30$, indicating that the approximation is robust with respect to $J$. \Cref{fig:E12}(A) shows that the interaction energy is negative definite and monotonically increasing in $d_{12}$, tending to zero as $d_{12} \to +\infty$. This implies, in particular, that the force exerted by each obstacle on the other is always attractive and vanishes at large distances, as expected. \Cref{fig:E12}(B) illustrates that the relative error decays rapidly to zero for large separations, but it deteriorates as the threshold distance is approached, {\em i.e.} as $d_{12} \to (\sqrt{2})^{+}$.
\medskip

\paragraph{{\bf N = 3.}}
In the case of three obstacles, we evaluate the Casimir energy fixing the positions of two obstacles and allowing the third to vary over $\mathbb{R}^3$, subject to \eqref{eq:ass2}. Without loss of generality, we may place the first two obstacles symmetrically along the $z$-axis. More precisely, for a given $a > 0$, we set
	\begin{equation}\label{eq:x3}
		y_1 = (0,0,-a/2)\,, \qquad y_2 = (0,0,a/2)\,, \qquad y_3 \in \R^3\,.
	\end{equation}
Exploiting the resulting rotational symmetry about the $z$-axis, we proceed to evaluate $\widetilde{\mathcal{E}}^{(\mathrm{int})}_{J}$ as a function of the cylindrical coordinates $(r,z) \in \R_+ \times \R$ of the third obstacle, suppressing the irrelevant azimuthal angular variable.

\begin{figure}[t!]
		\centering
		\begin{subfigure}{0.48\textwidth}
    			\centering
    			\includegraphics[width=\linewidth]{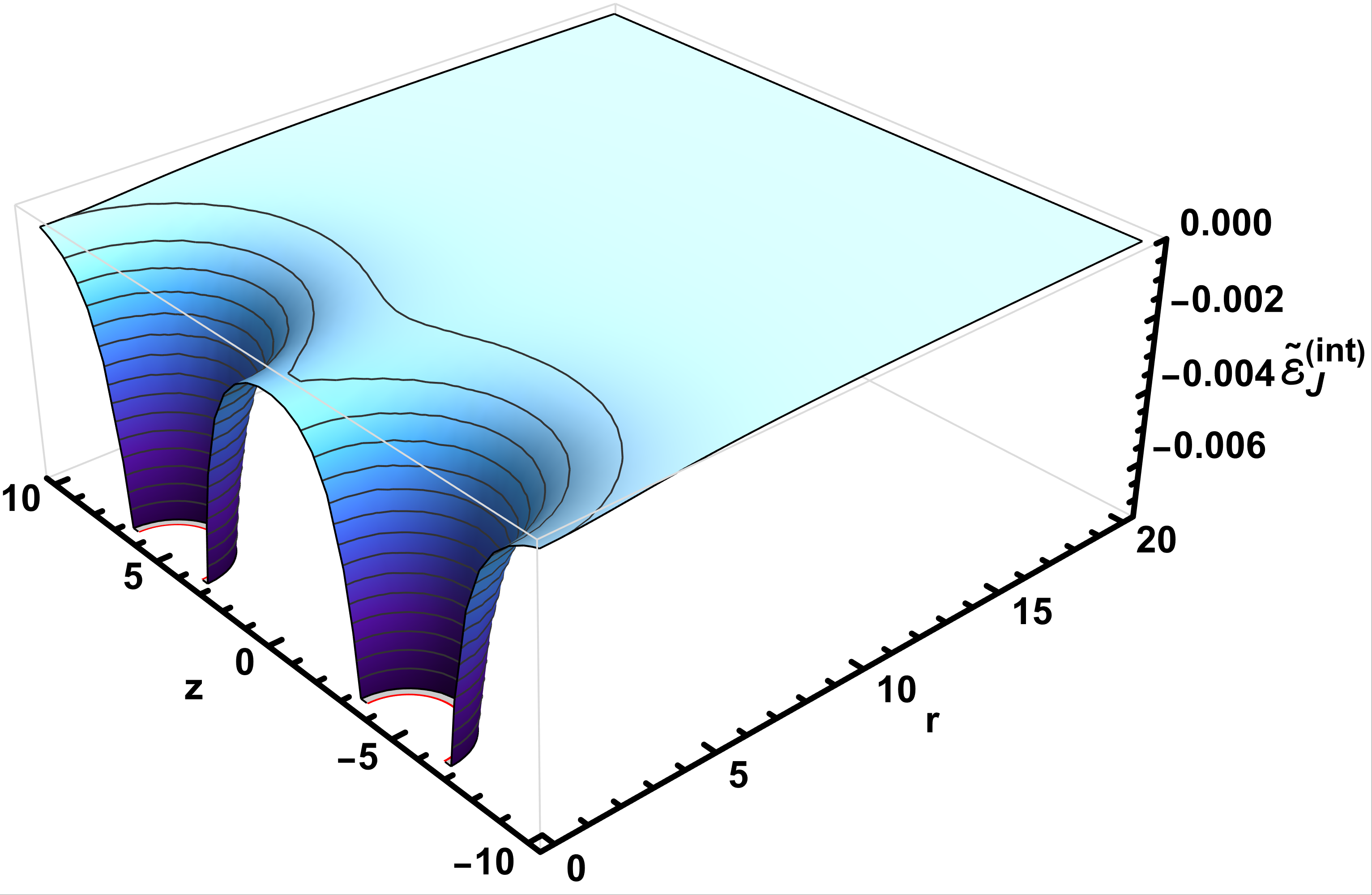}\label{fig:3a}
   			\caption{}
		\end{subfigure}
			\hfill
		\begin{subfigure}{0.48\textwidth}
   			\centering
   			\includegraphics[width=\linewidth]{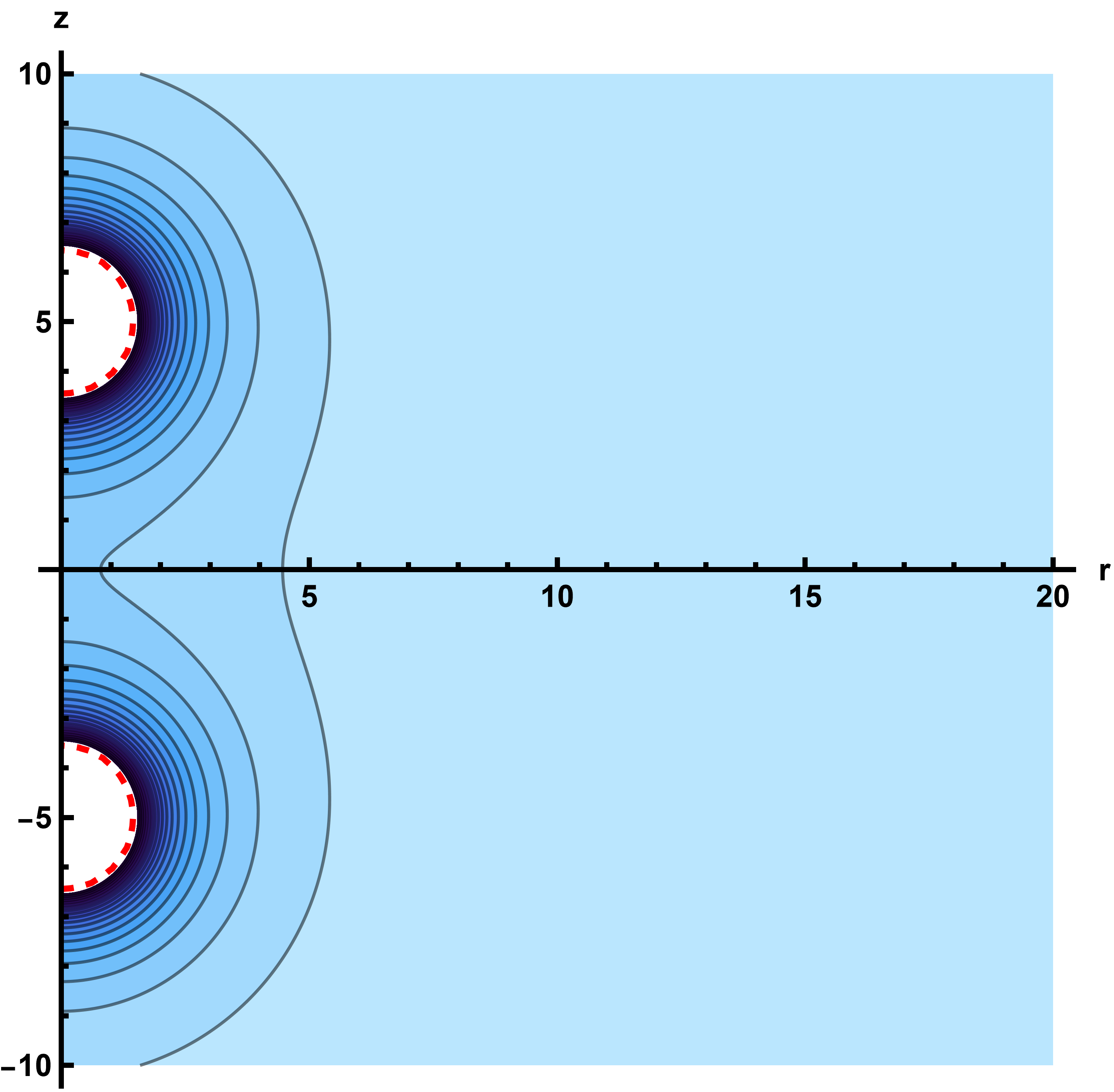}\label{fig:3b}
   			\caption{}
		\end{subfigure}
		\caption{
		$\widetilde{\mathcal{E}}^{(\mathrm{int})}_{J}$ for $N = 3$ and $J = 10$, with $(y_1,y_2,y_3)$ as in \eqref{eq:x3} and $a = 5$, plotted as a function of $(r,z)$.
		Panel~(A): surface plot. 
		Panel~(B): level curves. 
		The red dashed curves mark the cut-off imposed by \eqref{eq:ass2}.
		}
	\label{fig:E3}
	\end{figure}

\Cref{fig:E3} shows $\widetilde{\mathcal{E}}^{(\mathrm{int})}_{J}$, for $J = 10$ and $a = 5$, as a function of the coordinates $(r,z)$. The function is negative throughout the domain and exhibits two pronounced wells located at the positions of the fixed obstacles along the z-axis. Away from these regions, the surface becomes progressively flatter at larger distances, indicating that the interaction weakens with distance and is always attractive. The symmetry of the profile reflects the identical nature of the two centers. Moreover, the midpoint $y_3 = (0,0,0)$ appears to be the unique (unstable) equilibrium position.

As representative values of the relative error for $J=10$, we report
\begin{align*}
	\widetilde{\mathcal{R}}_{J}\Big|_{(r,z) = (3,5)} \leqslant 2 \times 10^{-2}\,, \qquad
	\widetilde{\mathcal{R}}_{J}\Big|_{(r,z) = (0,0)} \leqslant 6.2 \times 10^{-3}\,.
\end{align*}
These estimates confirm that the approximation achieves percent-level accuracy (or better) in the regions of interest, provided that the obstacles are sufficiently far apart, consistently with \eqref{eq:ass2}. Moreover, the relative error decreases rapidly with increasing distance from the centers $y_1$ and $y_2$. Although the accuracy can be systematically improved by taking larger values of J, this leads to a substantial increase in computational time.
\medskip

\paragraph{{\bf N = 4.}}
For illustrative purposes, we consider a configuration in which the first three obstacles are fixed at the vertexes of an equilateral triangle of side $b > 0$, lying in a given plane, while the position of the fourth obstacle varies within the same plane. For definiteness, we fix
\begin{equation}\label{eq:x4}
	y_1 = \left(\tfrac{b}{2}, - \tfrac{b}{2\sqrt{3}}, 0\right), \qquad
	y_2 = \left(-\tfrac{b}{2}, - \tfrac{b}{2\sqrt{3}}, 0\right), \qquad
	y_3 = \left(0, \tfrac{b}{\sqrt{3}}, 0\right), \qquad
	y_4 = (x,y,0)\,.
\end{equation}
With this choice, the centroid of the triangle $(y_1,y_2,y_3)$ lies at the origin of the chosen reference frame.

\begin{figure}[t!]
		\centering
		\begin{subfigure}{0.48\textwidth}
    			\centering
    			\includegraphics[width=\linewidth]{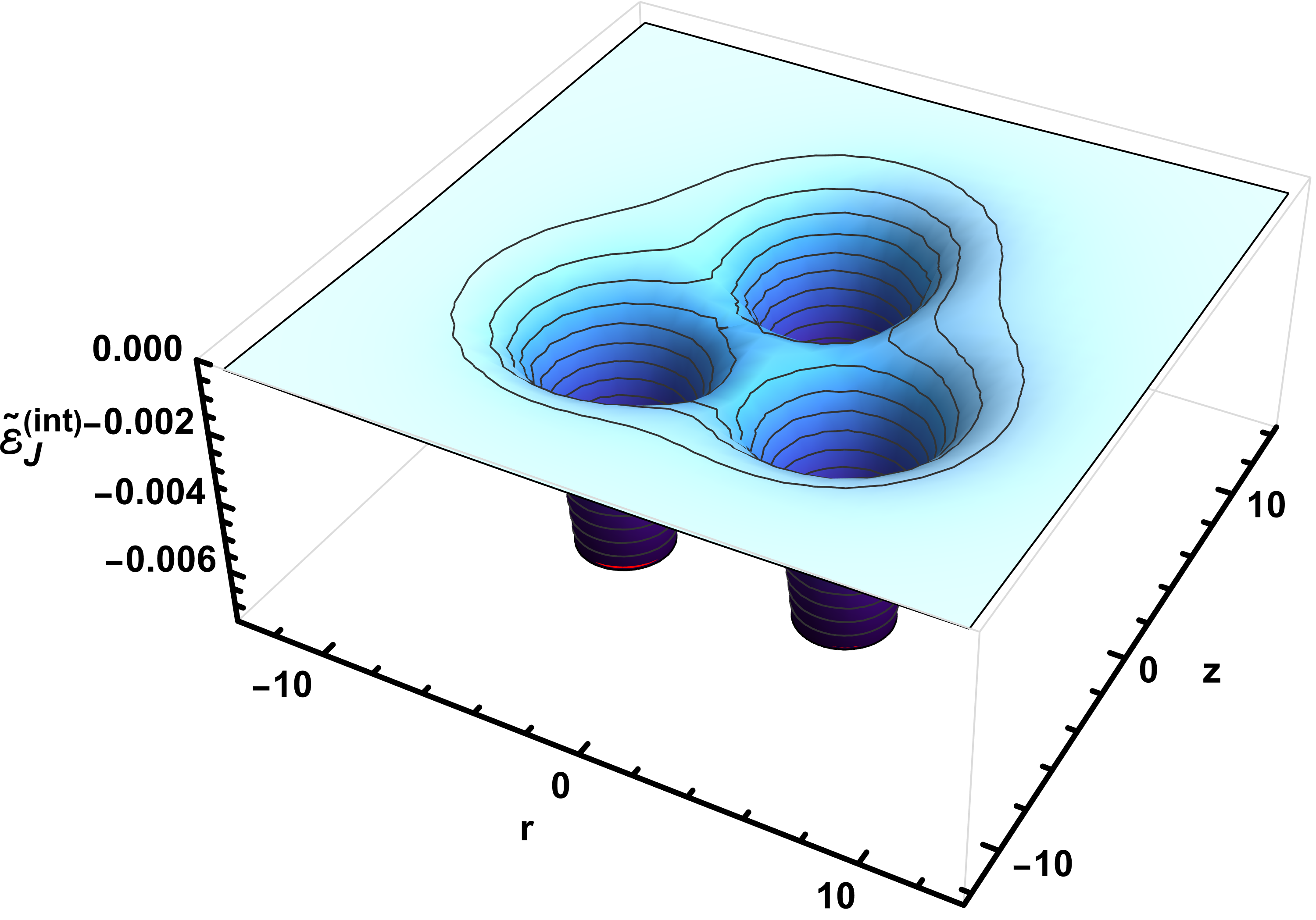}\label{fig:4a}
   			\caption{}
		\end{subfigure}
			\hfill
		\begin{subfigure}{0.48\textwidth}
   			\centering
   			\includegraphics[width=\linewidth]{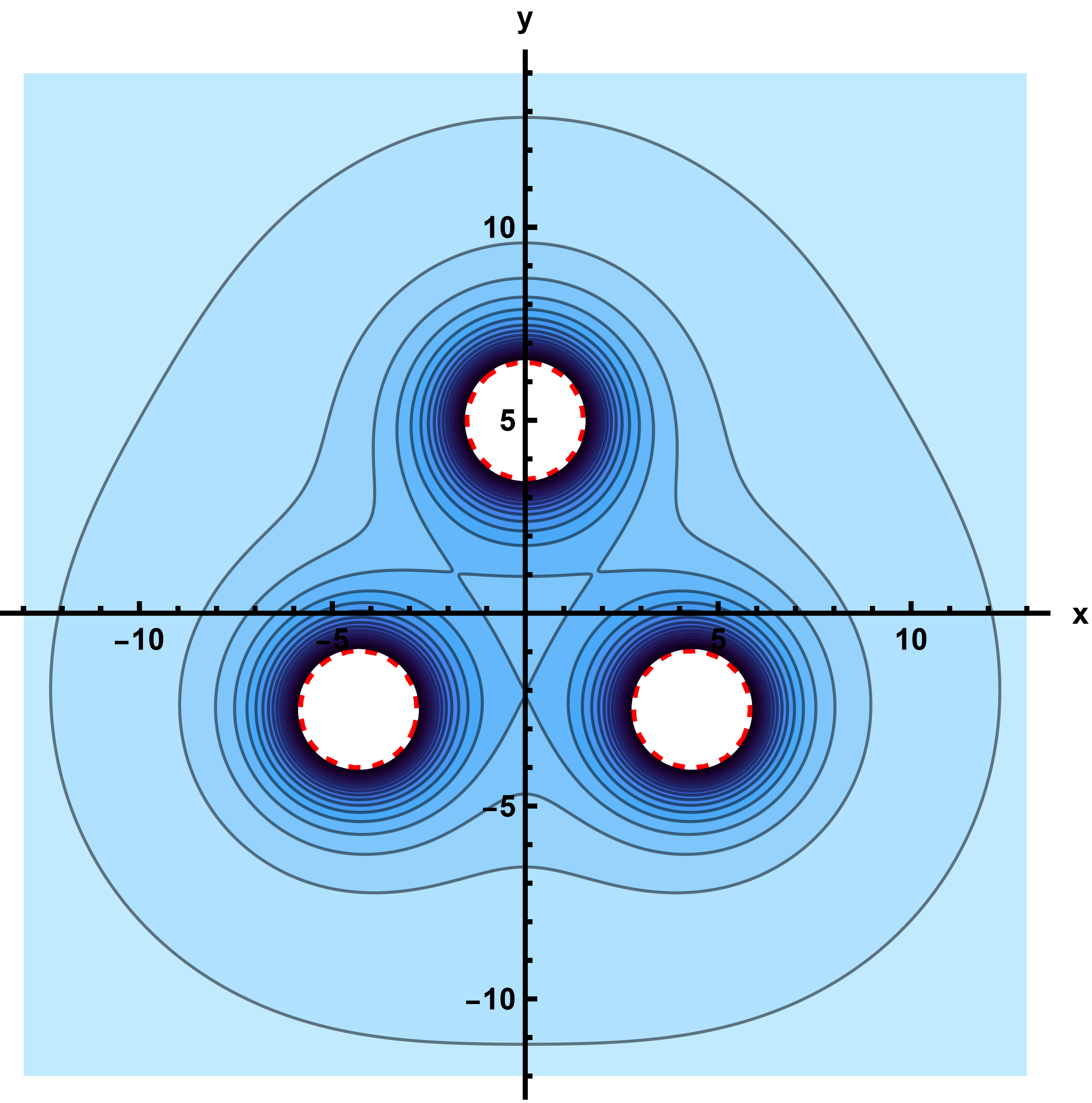}\label{fig:4b}
   			\caption{}
		\end{subfigure}
		\caption{
		$\widetilde{\mathcal{E}}^{(\mathrm{int})}_{J}$ for $N = 4$ and $J = 10$, with $(y_1,y_2,y_3)$ as in \eqref{eq:x4} and $b = 5\sqrt{3}$, plotted as a function of $(x,y)$.
		Panel~(A): surface plot. 
		Panel~(B): level curves. 
		The red dashed curves mark the cut-off imposed by \eqref{eq:ass2}.
		}
	\label{fig:E4}
	\end{figure}

\Cref{fig:E4} displays the approximate interaction energy $\widetilde{\mathcal{E}}^{(\mathrm{int})}_{J}$ for $J = 10$ and $b = 5\sqrt{3}$. 
The same qualitative behavior observed in the previously examined configurations clearly emerges here as well. In particular, the interaction energy decays rapidly as the distance from the obstacles increases. Furthermore, it exhibits pronounced wells in correspondence with the fixed positions of the first three obstacles, indicating the presence of attractive forces. The symmetry of the configuration is also reflected in the structure of the level sets of the interaction energy, which inherit the geometric arrangement of the fixed obstacles. Also in this case, the centroid $y_4 = (0,0,0)$ appears to be the unique (unstable) equilibrium position.

\section{Proofs}\label{sec:proofs}

\subsection{The relative resolvent trace}\label{subsec:restr}
Let us firstly establish the following basic result.

	\begin{lemma}\label{lemma:GaN}
		Assume \eqref{eq:ass1} and \eqref{eq:ass2}. Then, for any $z \in \C \setminus[0,+\infty)$, the matrix $\GaN^z$ is invertible and extends continuously to two distinct invertible matrices as $z$ approaches the cut $[0,+\infty)$ from above and from below.
		Furthermore, $(\GaN^z)^{-1}$ and its boundary extensions admit the absolutely convergent Neumann series representation
		\begin{equation}\label{eq:NSerG}
			\big(\GaN^z\big)^{-1} = \sum_{j = 0}^{+\infty} \big[(\VN)^{-1} \PN \big]^{j} (\VN)^{-1}\,,
		\end{equation}
		where the $N \times N$ matrices $\VN$ and $\PN$ are defined by
		\begin{equation}\label{eq:VNPN}
			\big[\VN\big]_{\mm\nn} := \left(\alpha_\nn - \frac{i\sqrt{z}}{4\pi}\right)\delta_{\mm\nn} \,, \qquad
			\big[\PN\big]_{\mm\nn} := \frac{e^{i \sqrt{z}\,|x_\mm-x_\nn|}}{4\pi|x_\mm-x_\nn|}\,(1-\delta_{\mm\nn})\,.
		\end{equation}
	\end{lemma}

\begin{proof}
First notice that, with the positions in \eqref{eq:VNPN}, the basic definition \eqref{eq:GaN} reads
\begin{equation*}
	\GaN^z = \VN - \PN \,.
\end{equation*}
Since $\alpha_\nn > 0$ for all $\nn \in \{1,\dots,N\}$ by \eqref{eq:ass1}, and since we are considering the branch of the square root with $\Im \sqrt{z} > 0$ for $z \in \C \setminus [0,+\infty)$ ($\Im \sqrt{z} = 0$ on the cut), it appears that the diagonal matrix $\VN$ has only entries with strictly positive real parts. Hence $\VN$ is invertible and we may factor
\begin{equation*}
	\GaN^z = \VN \big[\one - (\VN)^{-1} \PN \big]\,.
\end{equation*}
Invertibility of $\GaN^z$ and absolute convergence of the Neumann series \eqref{eq:NSerG} both follow once we establish
\begin{equation*}
	\big\|(\VN)^{-1} \PN\big\| < 1\,.
\end{equation*}
To this end, we estimate using the Hilbert-Schmidt norm:
\begin{multline*}
	\big\|(\VN)^{-1} \PN\big\|^2 
	\leqslant \big\|(\VN)^{-1} \PN \big\|_{\mathrm{HS}}^2
	= \tx \sum_{\mm,\nn = 1}^N \big| \big[(\VN)^{-1} \PN \big]_{\mm\nn}\big|^2 
	\\
	= \tx \sum_{\mm,\nn = 1}^N \Big|\Big(\alpha_\mm - \tfrac{i\sqrt{z}}{4\pi}\Big)^{-1} \frac{e^{i \sqrt{z}\,|x_\mm-x_\nn|}}{4\pi|x_\mm-x_\nn|}\,(1-\delta_{\mm\nn}) \Big|^2
	= \sum_{\mm \neq \nn} \frac{1}{(4\pi \alpha_\mm + \Im \sqrt{z})^2 + (\Re \sqrt{z})^2}\, \frac{e^{- 2\, \Im \sqrt{z} \,|x_\mm - x_\nn|}}{|x_\mm - x_\nn|^2} \\
	\leqslant \bigg(\sum_{\mm \neq \nn} \frac{1}{(4\pi \alpha_\mm)^2 |x_\mm - x_\nn|^2}\bigg)\; \frac{1}{1+\left(\tfrac{\Re \sqrt{z}}{4\pi \min_\mm \alpha_\mm}\right)^2} < 1 \,,
\end{multline*}
where the last inequality follows directly from \eqref{eq:ass2}.
Finally, the extensions of $\Gamma_N^z$ and $(\Gamma_N^z)^{-1}$ to both sides of the cut are derived noting that all the quantities involved depend continuously on $\Im \sqrt{z} \geqslant 0$.
\end{proof}

Next, recall that $R(z;\LN) - R(z;\Lz)$ is a finite-rank operator, hence of trace class, for any $z \in \C \setminus \big( \sigma(\LN) \cup \sigma(\Lz)\big)$. Using the explicit expression \eqref{eq:GN} and taking into account that $\Gz^z \in L^2(\R^3)$, by elementary arguments we obtain
	\begin{align*}
		r(z;\LN,\Lz) = \Tr \big[ R(z;\LN) - R(z;\Lz) \big]
		= \sum_{\mm,\nn=1}^N \big[\GaN^{z}\big]^{-1}_{\mm\nn} \int_{\R^3} dx\; \Gz^{z}(x - x_\mm)\,\Gz^{z}(x_\nn - x) \,.
	\end{align*}
Using the Fourier transform representation of the free Green function
	\begin{equation*}
		\Gz^{z}(x) = \frac{1}{(2\pi)^3}\!\int_{\R^3} dp\;\frac{e^{i p \cdot x}}{|p|^2-z} \,,
	\end{equation*}
we compute
	\begin{align*}
		\int_{\R^3} dx\; \Gz^{z}(x - x_\mm)\,\Gz^{z}(x_\nn - x) = \frac{1}{(2\pi)^3} \int_{\R^3} dp\;\frac{e^{-i p \cdot(x_\mm - x_\nn)}}{( |p|^2-z)^2}\,.
	\end{align*}
Passing to polar coordinates and applying the residue theorem, with $\Im\sqrt{z} >0$, we get
	\begin{multline*}
		\frac{1}{(2\pi)^3} \int_{\R^3} dp\;\frac{1}{(|p|^2-z)^2}
		= \frac{1}{2\pi^2} \int_{0}^\infty dr\;\frac{r^2}{(r^2-z)^2}
		\\
		= \frac{1}{4\pi^2} \int_{-\infty}^\infty dr\;\frac{r^2}{(r^2-z)^2}
		= \frac{i}{2\pi}\, \Res\left(\frac{r^2}{(r^2-z)^2};\sqrt{z}\right)
		= \frac{i}{8\pi \sqrt{z}}\,,
	\end{multline*}
for $\mm = \nn$, and 
	\begin{multline*}
		\frac{1}{(2\pi)^3} \int_{\R^3} dp\;\frac{e^{-i p \cdot (x_\mm - x_\nn)}}{(|p|^2-z)^2}
		= \frac{1}{(2\pi)^2} \int_0^\infty dr\,r^2 \;\frac{1}{(r^2-z)^2} \int_0^\pi d\theta\,\sin\theta \;e^{-i r |x_\mm - x_\nn| \cos\theta}
		\\
		= \frac{1}{2\pi^2 |x|} \int_0^\infty dr\,\frac{r\,\sin(r |x_\mm - x_\nn|)}{(r^2-z)^2} 
		= \frac{1}{4i\pi^2 |x_\mm - x_\nn|} \int_{-\infty}^\infty dr \;\frac{r}{(r^2-z)^2}\,e^{i r |x|}
		\\
		= \frac{1}{2\pi |x_\mm - x_\nn|}\, \Res\left( \frac{r}{(r^2-z)^2}\,e^{i r |x|}; \sqrt{z}\right)
		= \frac{i}{8\pi \sqrt{z}}\,e^{i \sqrt{z}|x_\mm - x_\nn|}\,.
	\end{multline*}
for $\mm \neq \nn$. Combining the above results, we obtain
	\begin{equation}\label{eq:rzexp}
		r(z;\LN,\Lz)
		= \frac{i}{8 \pi \sqrt{z}} \sum_{\mm,\nn=1}^N \big[\GaN^{z}\big]^{-1}_{\mm\nn} \,e^{i \sqrt{z}|x_\mm-x_\nn|}\,.
	\end{equation}
where $\GaN^{z}$ is defined as in \eqref{eq:GaN} and the existence of its inverse is ensured by \cref{lemma:GaN}. 
Notice that the above expression for $r(z;\LN,\Lz)$ depends on the spectral parameter $z$ only through its square root $\sqrt{z}$, with the stated determination.
\medskip

We are now ready to prove \cref{prop:rz}.

\begin{proof}[Proof of \cref{prop:rz}]
In view of the characterization \eqref{eq:sigLN} of $\sigma(\LN)$, the claim \eqref{eq:sigLNp} follows straightforwardly from \cref{lemma:GaN}.
Next, we derive the asymptotic expansions \eqref{eq:asyrz0} and \eqref{eq:asyrzinf}.

{\sl i)} Referring again to \eqref{eq:GaN}, we split
	\begin{equation*}
		\GaN^{z} = \GaN^0 - \QN\,,
	\end{equation*}
with
	\begin{equation}
		\big[\GaN^0\big]_{\mm\nn} = \alpha_\nn\,\delta_{\mm\nn} - \tfrac{1}{4\pi|x_\mm-x_\nn|}\,(1-\delta_{\mm\nn})\,,
		\qquad
		\big[\QN\big]_{\mm\nn} := \tfrac{i\sqrt{z}}{4\pi} \left[\delta_{\mm\nn} + \tfrac{e^{i\sqrt{z}|x_\mm - x_\nn|} - 1}{i\sqrt{z} |x_\mm - x_\nn|}\,(1-\delta_{\mm\nn}) \right] . \label{eq:GaN0Q}
	\end{equation}
Under the assumptions \eqref{eq:ass1} and \eqref{eq:ass2}, \cref{lemma:GaN} ensures that $\GaN^0$ is invertible. On top of that, an elementary computation shows that $\|\QN\| < c\,|z|^{1/2}$ for some finite constant $c > 0$ and all $|z|$ small enough. As a consequence, we have $\|(\GaN^0)^{-1} \QN\| < 1$ for all $z \in \C$ sufficiently close to $0$, and we can write
	\begin{equation*}
		\big(\GaN^{z}\big)^{-1} = \big[\one -(\GaN^0)^{-1} \QN\big]^{-1} (\GaN^0)^{-1}
		= \tx \sum_{j = 0}^{+\infty} \big[(\GaN^0)^{-1} \QN\big]^{j} (\GaN^0)^{-1}\,.
	\end{equation*}
Combining this relation with the explicit expression \eqref{eq:rzexp} and the analytic dependence of $\QN$ on $\sqrt{z}$
ultimately accounts for \eqref{eq:asyrz0}.

{\sl ii)} In this case, we rephrase \eqref{eq:GaN} as 
	\begin{equation*}
		\GaN^{z} = -\frac{i \sqrt{z}}{4\pi}\, \big( \one - \SN \big)\,,
	\end{equation*}
where
	\begin{equation*}
		\big[\SN\big]_{\mm\nn} := \frac{4\pi}{i\sqrt{z}} \left[\alpha_\nn\,\delta_{\mm\nn} - \tfrac{e^{i \sqrt{z}\,|x_\mm-x_\nn|}}{4\pi|x_\mm-x_\nn|}\,(1-\delta_{\mm\nn})\right] .
	\end{equation*}
Recalling once more that $\Im \sqrt{z} > 0$, it is easy to see that $\|\SN\| < 1$ for any $z$ sufficiently large. Thus, we can consider the absolutely convergent series expansion
	\begin{equation}
		\big(\GaN^{z}\big)^{-1} = -\frac{4\pi}{i \sqrt{z}}\, \big( \one - \SN \big)^{-1}
		= -\frac{4\pi}{i \sqrt{z}} \tx\sum_{j = 0}^{+\infty}\, (\SN)^j\,.
	\end{equation}
To say more, considering that the off-diagonal terms vanish with exponential rate as $z \to \infty$, we have
	\begin{equation*}
		\big[(\SN)^j\big]_{\mm\nn} = \left(\frac{4\pi \alpha_\nn}{i\sqrt{z}}\right)^j \delta_{\mm\nn}+ O\big(z^{-j/2}\,e^{i\sqrt{z} \min_{\mm\neq \nn}|x_\mm - x_\nn|}\big) \,, \qquad \mbox{for all $j \in \Nz$}\,.
	\end{equation*}
Substituting the above relations into the explicit expression \eqref{eq:rzexp} for $r(z;\LN,\Lz)$ yields
	\begin{equation*}
		r(z;\LN,\Lz)
		\sim -\frac{1}{2 z} \sum_{j = 0}^{+\infty} \big(\tx \sum_{\nn=1}^N \alpha_\nn^j\big) \left(\tfrac{4\pi}{i\sqrt{z}}\right)^j ,
	\end{equation*}
which proves \eqref{eq:asyrzinf}.
\end{proof}

\subsection{The relative spectral density}\label{subsec:respden}
We now analyze the relative spectral density $e(v;\LN,\Lz)$ defined in \eqref{eq:defee}. Recalling the explicit expression \eqref{eq:rzexp} for the relative resolvent trace, we observe that it depends on the spectral parameter $z \in \C \setminus [0,+\infty)$ only through the square root $\sqrt{z}$. Furthermore, with the established determination of the square root, we have
	\begin{equation*}
		\sqrt{v^2 e^{i\varepsilon}} \to v\,,	\qquad
		\sqrt{v^2 e^{i(2\pi - \varepsilon)}} \to - v\,, \qquad
		\mbox{as $\varepsilon \to 0^+$}\,.
	\end{equation*}
Hence, we readily obtain
	\begin{equation}\label{eq:eexp}
		e(v;\LN,\Lz) = \frac{1}{4 \pi^2} \, \Re \left( \sum_{\mm,\nn=1}^N \big[\GaN^{+,v^2}\big]^{-1}_{\mm\nn} \;e^{i v |x_\mm-x_\nn|}\right),
	\end{equation}
where, cf. \eqref{eq:GaN},
	\begin{equation}\label{eq:GaNp}
		\big[\Gamma_{N}^{+,v^2}\big]_{\mm\nn}
		:= \lim_{\varepsilon \to 0^+} \big[\Gamma_{N}^{v^2 e^{i\varepsilon}}\big]_{\mm\nn}
		= \left(\alpha_\mm - \frac{i v}{4\pi}\right)\delta_{\mm\nn} - 
\frac{e^{i v\,|x_\mm-x_\nn|}}{4\pi|x_\mm-x_\nn|}\,(1-\delta_{\mm\nn})\,.
	\end{equation}
This is precisely the continuous extension of $\Gamma_N^z$ to the branch cut $[0,+\infty)$ from above, whose existence and invertibility for all $v > 0$ were previously established in \cref{lemma:GaN}. The analogous extension from below, namely $\Gamma_{N}^{-,v^2}$, satisfies $\Gamma_{N}^{-,v^2} = \overline{\Gamma_{N}^{+,v^2}}$, which explains the appearance of the real part in \eqref{eq:eexp}.

\begin{proof}[Proof of \cref{prop:ee}]
The smoothness of $e(v;\LN,\Lz)$ follows directly from \eqref{eq:eexp}, \eqref{eq:GaNp} and \cref{lemma:GaN}. We now discuss its asymptotic behaviors for $v \to 0^+$ and $v \to + \infty$.

{\sl i)} The claim can be derived by arguing as in the proof of claim {\sl (i)} of \cref{prop:rz}. More precisely, fix $\GaN^0$ as in \eqref{eq:GaN0Q} and set
	\begin{equation*}
		\big[\QNp\big]_{\mm\nn} := \tfrac{i v}{4\pi} \left[\delta_{\mm\nn} + \tfrac{e^{i v |x_\mm - x_\nn|} - 1}{iv |x_\mm - x_\nn|}\,(1-\delta_{\mm\nn}) \right] .
	\end{equation*}
Notice that $\QNp = \lim_{\varepsilon \to 0^+} Q_N^{v^2 e^{i \varepsilon}}$.	
Then, recalling that the assumptions \eqref{eq:ass1} and \eqref{eq:ass2} ensure the uniform invertibility of $\GaN^0$, and noting that $\big\|(\GaN^0)^{-1} \QNp\big\| < 1$ for $v$ small enough, we may write
	\begin{equation}\label{eq:GaNpv0}
		\big(\GaN^{+,v^2}\big)^{-1} = \tx \sum_{j = 0}^{+\infty} \big[(\GaN^0)^{-1} \QNp\big]^{j} (\GaN^0)^{-1}\,.
	\end{equation}
The thesis \eqref{eq:asyev0} follows by inserting \eqref{eq:GaNpv0} into \eqref{eq:eexp} and observing that $\QNp$ depends analytically on $v$. In particular, the fact that only even powers of $v$ appear is a consequence of the analyticity of all functions with respect to the variable $i v$, together with the extraction of the real part in \eqref{eq:eexp}.

{\sl ii)} The argument parallels that in the proof of claim {\sl (ii)} of \cref{prop:rz}, with one important difference. Setting 
	\begin{equation*}
		\big[\SNp\big]_{\mm\nn} := \frac{4\pi}{i v} \left[\alpha_\nn\,\delta_{\mm\nn} - \tfrac{e^{i v\,|x_\mm-x_\nn|}}{4\pi|x_\mm-x_\nn|}\,(1-\delta_{\mm\nn})\right] \,,
	\end{equation*}
so that $\SNp = \lim_{\varepsilon \to 0^+} S_N^{v^2 e^{i \varepsilon}}$, we obtain
	\begin{equation*}
		\big[\GaN^{+,v^2}\big]^{-1}
		= -\frac{4\pi}{i v} \tx\sum_{j = 0}^{+\infty}\, (\SNp)^j\,.
	\end{equation*}
	Unlike in the analysis of $r(z;\LN,\Lz)$, here the off-diagonal terms do not decay exponentially, although they remain uniformly bounded for $v\in(0,+\infty)$. Upon evaluating the real part in \eqref{eq:eexp}, they thus produce oscillating contributions in the asymptotic expansion of $e(v;\LN,\Lz)$, leading to \eqref{eq:asyeving}. It appears that the non-oscillating terms descend from purely diagonal contributions. More precisely, one sees that
	\begin{multline*}
		e(v;\LN,\Lz) = \frac{1}{4 \pi^2} \, \Re \sum_{\nn=1}^N  \left(-\frac{4\pi}{i v} \tx\sum_{j = 0}^{+\infty}\, [\SNp]_{\nn\nn}^j\right) + \mbox{oscillating terms}\\
		= - \,\frac{1}{\pi} \, \Im \left( \tx\sum_{j = 0}^{+\infty} \sum_{\nn=1}^N \frac{(-4\pi i\, \alpha_\nn)^j}{v^{j+1}} \right) + \mbox{oscillating terms}\,.
	\end{multline*}
This yields \eqref{eq:gj}, upon observing that only terms with odd values of $j$ contribute to the imaginary part and relabelling the summation index accordingly. We also note that, as $v \to +\infty$,	
	\begin{multline*}
		 \ee(v;\LN,\Lz) 
		 = -\tfrac{1}{\pi v} \, \Im \tx\sum_{\mm,\nn=1}^N\left( \delta_{\mm\nn} + [\SNp]_{\mm\nn}  + O(v^{-2}) \right) e^{i v |x_\mm-x_\nn|} \\		 
		 = - \tfrac{1}{\pi v} \, \Im \left[ N + \tfrac{4\pi}{i v} \tx\left(\sum_{\nn=1}^N  \alpha_\nn - \sum_{\mm \neq \nn} \tfrac{e^{2i v\,|x_\mm-x_\nn|}}{4\pi|x_\mm-x_\nn|}\right)  + O(v^{-2}) \right] \\
		 = \tfrac{4}{v^2}\! \tx\left(\sum_{\nn=1}^N  \alpha_\nn - \sum_{\mm \neq \nn} \tfrac{\cos(2 v\,|x_\mm-x_\nn|)}{4\pi|x_\mm-x_\nn|}\right) + O(v^{-3})\,,
		\end{multline*}
	which explains the absence of the $v^{-1}$ term in \eqref{eq:asyeving} and further accounts for \eqref{eq:eevinf}.
\end{proof}

\subsection{Analytic continuation of the relative zeta function}\label{subsec:acrelz}
The fact that the definition \eqref{eq:zetamull} of $\zeta(s;\LN,\Lz)$ is well posed for all $s \in \C$ with $0 < \Re s <1/2$ follows from \cref{prop:rz} and the arguments in \cite{SZ09} (see, especially, the comments below Eq. (6) therein). To construct the analytic continuation to the entire complex plane, we refer to the representation \eqref{eq:zetaev} and split, for some $\vz > 0$ fixed arbitrarily,
	\begin{equation}\label{eq:Zsplit}
		\zeta(s;\LN,\Lz) = \zeta_{\vz}^{\IR}(s;\LN,\Lz) + \zeta_{\vz}^{\UV}(s;\LN,\Lz)\,,
	\end{equation}
where
	\begin{eqnarray}
		& \displaystyle{\zeta_{\vz}^{\IR}(s;\LN,\Lz) := \int_0^{\vz}\! dv\;v^{-2s}\,\ee(v;\LN,\Lz)\,,} \label{eq:zetaIR}\\
		& \displaystyle{\zeta_{\vz}^{\UV}(s;\LN,\Lz) := \int_{\vz}^{\infty}\! dv\;v^{-2s}\,\ee(v;\LN,\Lz)\,.} \label{eq:zetaUV}
	\end{eqnarray}
The labels $\IR$ and $\UV$ indicate that $\zeta_{\vz}^{\IR}(s;\LN,\Lz)$ and $\zeta_{\vz}^{\UV}(s;\LN,\Lz)$ capture, respectively, the infra-red and ultra-violet contributions. 
The meromorphic extensions of these two functions are treated separately in the forthcoming \cref{lemma:ZIR,lemma:ZUV}, which together ultimately prove \cref{thm:zLNLz}.
	
	\begin{lemma}\label{lemma:ZIR}
		Assume \eqref{eq:ass1} and \eqref{eq:ass2}. Then, the map $s \mapsto \zeta_{\vz}^{\IR}(s;\LN,\Lz)$ is well-defined and analytic in the half-plane $\{\Re\, s < 1/2\}$, and it possesses a meromorphic continuation to the whole complex plane, with possible simple poles at $s \in 1/2+\Nz$. Moreover, for any $\JIR \in \Nz$ and $\Re s < \JIR+3/2$, one has
		\begin{equation}\label{eq:ZIRac}
			\zeta_{\vz}^{\IR}(s;\LN,\Lz) = \sum_{j = 0}^{\JIR}  \frac{f_j\,\vz^{2j + 1 - 2s}}{2j +1 - 2s} + \int_0^{\vz}\! dv\;v^{-2s}\, \Big[\ee(v;\LN,\Lz) - {\tx \sum_{j = 0}^{\JIR}}\, f_j \,v^{2j} \Big]\,,
		\end{equation}
		where the coefficients $f_j$ are those introduced in \eqref{eq:asyev0} and the integral is absolutely convergent.
	\end{lemma}

\begin{proof}
By \cref{prop:ee}, there exists a constant $c_0 > 0$ such that $|v^{-2s} \ee(v;\LN,\Lz)| \leqslant c_0 \,v^{-2\varsigma_0}$ for all $v \in (0,\vz)$ and $\Re s < \varsigma_0$, for any fixed $\varsigma_0 < 1/2$. Since $v^{-2\varsigma_0} \in L^1(0,\vz)$, the integral in \eqref{eq:zetaIR} is absolutely convergent. On top of that, for each $v > 0$, the function $s \in \C \mapsto v^{-2s} \ee(v;\LN,\Lz)$ is entire. Hence, for any closed path $\gamma \subset \{\Re s < 1/2\}$, by Fubini's theorem we infer
\begin{align*}
	\int_{\gamma} ds \int_0^{\vz}\! dv\;v^{-2s}\,\ee(v;\LN,\Lz)
	= \int_0^{\vz}\! dv  \left( \int_{\gamma} ds\; v^{-2s} \right)\ee(v;\LN,\Lz) = 0\,, 
\end{align*}
where the inner integral in the second expression vanishes by Cauchy's theorem. Analyticity of the map $s \mapsto \zeta_{\vz}^{\IR}(s;\LN,\Lz)$ for $\Re s < 1/2$ then follows by Morera's theorem \cite[Thm. 3.1.4]{GK97}.

To obtain \eqref{eq:ZIRac}, we substitute the asymptotic expansion \eqref{eq:asyev0} of $\ee(v;\LN,\Lz)$, truncated at order $\JIR \in \Nz$, into \eqref{eq:zetaIR} and evaluate the resulting terms explicitly:
\begin{multline*}
	\zeta_{\vz}^{\IR}(s;\LN,\Lz)
	= \sum_{j = 0}^{\JIR} f_j \int_0^{\vz}\! dv\;v^{-2s + 2j} + \int_0^{\vz}\! dv\;v^{-2s} \Big[\ee(v;\LN,\Lz) - {\tx \sum_{j = 0}^{\JIR}} f_j \,v^{2j} \Big] \\
	= \sum_{j = 0}^{\JIR}  \frac{f_j\,\vz^{2j + 1 - 2s}}{2j +1 - 2s} + \int_0^{\vz}\! dv\;v^{-2s}\, \Big[\ee(v;\LN,\Lz) - {\tx \sum_{j = 0}^{\JIR}} f_j \,v^{2j} \Big] \,.
\end{multline*}
Although derived under the assumption $\Re s < 1/2$, the latter expression provides the analytic continuation of $\zeta_{\vz}^{\IR}(s;\LN,\Lz)$ to a larger region. In fact, the first sum defines a meromorphic function on $\C$, with possible simple poles at $s = j+1/2$ for $j = 0,\dots,\JIR$, while the remainder integral identifies an analytic function of $s$ for $\Re s < \JIR + 3/2$, by arguments analogous to those described above.
\end{proof}

	\begin{lemma}\label{lemma:ZUV}
		Assume \eqref{eq:ass1} and \eqref{eq:ass2}. Then, the map $s \mapsto \zeta_{\vz}^{\UV}(s;\LN,\Lz)$ is well-defined and analytic in the half-plane $\{\Re\, s > - 1/2\}$, and it possesses a meromorphic continuation to the whole complex plane, with possible simple poles at $s \in - 1/2 -\Nz$. Moreover, for any $\JUV \in \Nz$ and $\Re s > - \JUV/2-1$, one has
		\begin{multline}\label{eq:ZUVac}
		\zeta_{\vz}^{\UV}(s;\LN,\Lz) = \sum_{j = 0}^{\JUV} \frac{g_j\, \vz^{-2s-j-1}}{2s+j+1}
			- \sum_{j = 0}^{\JUV} \sum_{k = 0}^{K_j} \sum_{l = 0}^{\JUV - j}  \vz^{-2s -j-l-2} \,(2s+j+2)_l \,
			\Re \bigg( \frac{h_{j,k}\, e^{i \vz d_{j,k}}}{(i d_{j,k})^{l + 1}} \bigg) \\
			+ \sum_{j = 0}^{\JUV}\sum_{k = 0}^{K_j} \, (2s+j+2)_{\JUV - j+1} \int_{\vz}^{\infty}\! dv\; v^{-2s -\JUV -3}  \,\Re \bigg( \frac{h_{j,k}\,e^{i d_{j,k} v}}{(i  d_{j,k})^{\JUV - j + 1}}	
			 \bigg) \\			 
		+ \int_{\vz}^{\infty}\! dv\;v^{-2s} \left [\ee(v;\LN,\Lz) - \tx\sum_{j = 0}^{\JUV} v^{-j-2} \bigg( g_j + \Re\Big(\sum_{k = 0}^{K_j} h_{j,k}\,e^{i d_{j,k} v}\Big) \bigg) \right] .
	\end{multline}
		where $(a)_0 = 1$, $(a)_l = a (a+1) \,\dots (a+l-1)$ for $l \in \Nz$ are the Pochhammer symbols, the coefficients $g_j,h_{j,k},d_{j,k}$ are those introduced in \eqref{eq:asyeving}, and all the integrals are absolutely convergent.
	\end{lemma}

\begin{proof}
The well-posedness and analyticity of the integral expression \eqref{eq:zetaUV} can be established by arguments analogous to those reported in the first part of the proof of \cref{lemma:ZIR}. In this connection, notice that claim {\sl (ii)} in \cref{prop:ee} gives $|v^{-2s} \ee(v;\LN,\Lz)| \leqslant c'_0 \,v^{-2\varsigma_0-2} \in L^1(\vz,+\infty)$ for some $c'_0 > 0$ and all $\Re s > \varsigma'_0$, for any fixed $\varsigma'_0 > -1/2$.

 We now turn to the derivation of \eqref{eq:ZUVac}. Substituting the asymptotic expansion \eqref{eq:asyeving} of $\ee(v;\LN,\Lz)$, truncated at order $\JUV \in \Nz$, into \eqref{eq:zetaUV} we obtain
	\begin{multline}\label{eq:zUV3}
		\zeta_{\vz}^{\UV}(s;\LN,\Lz) = \sum_{j = 0}^{\JUV}g_j \int_{\vz}^{\infty}\! dv\;v^{-2s-j-2}\\
			+ \frac{1}{2}\sum_{j = 0}^{\JUV}\sum_{k = 0}^{K_j}\bigg(h_{j,k} \int_{\vz}^{\infty}\! dv\;v^{-2s -j-2} e^{i d_{j,k} v} + h^*_{j,k} \int_{\vz}^{\infty}\! dv\;v^{-2s -j-2} e^{-i d_{j,k} v}\bigg) \\
		+ \int_{\vz}^{\infty}\! dv\;v^{-2s} \left [\ee(v;\LN,\Lz) - \tx\sum_{j = 0}^{\JUV} v^{-j-2} \bigg( g_j + \Re\Big(\sum_{k = 0}^{K_j} h_{j,k}\,e^{i d_{j,k} v}\Big) \bigg) \right] .
	\end{multline}
Of course, for any $j = 0,\dots,\JUV$, we have
	\begin{equation*}
		\int_{\vz}^{\infty}\! dv\;v^{-2s-j-2} = \frac{\vz^{-2s-j-1}}{2s+j+1}\,,
	\end{equation*}
which identifies a meromorphic function of $s \in \C$ with a simple pole in $s = -(j+1)/2$.
As for the oscillatory integrals in \eqref{eq:zUV3}, keeping in mind that $d_{j,k} >0$ and $\vz >0$, an inductive argument based on repeated integration by parts yields, for any $M \in \Nz$,
	\begin{multline*}
		\int_{\vz}^{\infty}\! dv\;v^{-2s -j-2} e^{\pm i d_{j,k} v}
		= \vz^{-2s -j-1} \int_{1}^{\infty}\! dw\;w^{-2s -j-2} e^{\pm i \vz d_{j,k} w} \\
		= \vz^{-2s -j-1} \left[- e^{\pm i \vz d_{j,k}} \sum_{l = 0}^M \frac{(2s+j+2)_l}{(\pm i \vz d_{j,k})^{l + 1}} + \frac{(2s+j+2)_{M+1}}{(\pm i \vz d_{j,k})^{M + 1}} \int_{1}^{\infty}\! dw\; w^{-2s -j-M-3} e^{\pm i \vz d_{j,k} w}\right].
	\end{multline*}
It is evident that the first addenda are entire for any $s \in \C$. On the other hand, it can be proved that the latter integral defines an analytic function of $s$ for $\Re s > -(j+M)/2-1$, by arguments analogous to those outlined above.
Similarly, on account of claim {\sl (ii)} in \cref{prop:ee}, the remainder integral in the last line of \eqref{eq:zUV3} is analytic for $\Re s > -\JUV/2 - 1$.
The thesis follows by combining these results, recalling that $g_j = 0$ for odd values of $j$ and fixing $M = \JUV - j$.
\end{proof}

\begin{proof}[Proof of \cref{thm:zLNLz}]
The thesis is derived using \eqref{eq:Zsplit} and \cref{lemma:ZIR,lemma:ZUV}. In particular, \eqref{eq:ZIRac} and \eqref{eq:ZUVac} show that $s = 0$ is not a pole of $\zeta(s;\LN,\Lz)$, as neither $\zeta_{\vz}^{\IR}(s;\LN,\Lz)$ nor $\zeta_{\vz}^{\UV}(s;\LN,\Lz)$ are singular there. On the other hand, $s = -1/2$ is a point of analyticity for $\zeta_{\vz}^{\IR}(s;\LN,\Lz)$ and simple pole for $\zeta_{\vz}^{\UV}(s;\LN,\Lz)$.
\end{proof}

\subsection{Computing the renormalized relative partition function}\label{subsec:relpart}
Exploiting the explicit structure of the imaginary time component of the Euclidean Klein-Gordon operator $\AN$, one can relate the analytic continuation of the relative zeta function $\zeta(s;\AN,\Az)$ at $s = 0$ to that of its purely spatial counterpart $\zeta(s;\LN,\Lz)$ at $s = -1/2$. More precisely, \cite[Prop. 3.1]{SZ09} yields
	\begin{gather*}
		\Rez_{s=0} \zeta(s;\AN,\Az) = - \beta \Reu_{s = -1/2} \zeta(s;\LN,\Lz)\,, \\
		\Rez_{s=0} \zeta'(s;\AN,\Az) = - \beta \Rez_{s = -1/2} \zeta(s;\LN,\Lz) - 2\beta(1-\log 2) \Reu_{s = -1/2}\zeta(s;\LN,\Lz) -2 \log\eta(\beta;\LN,\Lz) \,,
	\end{gather*}
where, for $k = 0,1$, $\displaystyle{\Rek_{s = s_*}} f(s)$ denotes the coefficient of $(s-s_*)^{-k}$ in the Laurent expansion of $f(s)$ at $s = s_* \in \C$, while $\log\eta(\beta;L,L_0)$ is defined as in \eqref{eq:logndef}.
Substituting the above relations into the definition \eqref{eq:lnZN} of the renormalized connected relative partition function, one obtains
	\begin{multline}\label{eq:ZZzeta}
		\big[\ln \ZeN - \ln \Zez\big]_{\ren}
		= \beta \big[\log(2\ell) - 1\big] \Reu_{s=-1/2} \zeta(s;\LN,\Lz) \\
		- \tfrac{\beta}{2} \Rez_{s = -1/2} \zeta(s;\LN,\Lz) - \log\eta(\beta;\LN,\Lz)\,.
	\end{multline}
We remark that this expression is well-defined, since we already established in \cref{thm:zLNLz} that the (analytic continuation of the) relative zeta function $\zeta(s;\LN,\Lz)$ has a simple pole at $s = -1/2$.

\begin{proof}[Proof of \cref{thm:ZZren}]
We consider the representation \eqref{eq:ZZzeta} and refer to the splitting \eqref{eq:Zsplit} of the relative zeta function $\zeta(s;L,L_0)$. On account of \cref{lemma:ZIR}, we have that the IR contribution $\zeta_{\vz}^{\IR}(s;\LN,\Lz)$ is analytic at $s = -1/2$, for any $\vz > 0$. In particular, we have
	\begin{align*}
		\Reu_{s = -1/2} \zeta_{\vz}^{\IR}(s;\LN,\Lz)  = 0\,, \qquad
		\Rez_{s = -1/2} \zeta_{\vz}^{\IR}(s;\LN,\Lz)  = \int_0^{\vz}\! dv\;v\; \ee(v;\LN,\Lz)\,.
	\end{align*}
On the other hand, the analytic continuation of the UV term $\zeta_{\vz}^{\UV}(s;\LN,\Lz)$ has a simple pole at $s = -1/2$. Using the representation \eqref{eq:ZUVac} with $\JUV = 0$, we infer 
	\begin{equation*}
		\Reu_{s = -1/2} \zeta_{\vz}^{\UV}(s;\LN,\Lz) = \frac{1}{2}\,g_0 \,,
	\end{equation*}
	\begin{multline*}
		\Rez_{s = -1/2} \zeta_{\vz}^{\UV}(s;\LN,\Lz) = - g_0 \log \vz
			- \sum_{k = 0}^{K_0} \Re \bigg( \frac{h_{0,k}\, e^{i \vz d_{0,k}}}{i \vz d_{0,k}} \bigg)
			+ \sum_{k = 0}^{K_0} \,\vz \int_{\vz}^{\infty}\! dv\; v^{-2}  \,\Re \bigg( \frac{h_{0,k}\,e^{i d_{0,k} v}}{(i \vz d_{0,k})}	
			 \bigg) \\			 
		+ \int_{\vz}^{\infty}\! dv\;v \left [\ee(v;\LN,\Lz) - \tx v^{-2} \bigg( g_0 + \Re\Big(\sum_{k = 0}^{K_0} h_{0,k}\,e^{i d_{0,k} v}\Big) \bigg) \right].
	\end{multline*}
The thesis follows readily from the above results, recalling the explicit expressions \eqref{eq:gj} and \eqref{eq:eevinf}.
\end{proof}

\subsection{Asymptotic expansions of the Dedekind eta function}\label{subsec:ded}
We now study the behavior of the Dedekind eta function $\eta(\beta;\LN,\Lz)$, defined in \eqref{eq:logndef}, in both the low-temperature limit $\beta \to +\infty$ and in the high-temperature regime $\beta \to 0^+$.

As a preliminary, we observe that
	\begin{equation}\label{eq:dlogndef}
		\beta\,\partial_\beta \log \eta(\beta;\LN,\Lz) = \int_0^\infty\! dv\; \frac{\beta v}{e^{\beta v}-1}\,\ee(v;\LN,\Lz)\,.
	\end{equation}
This identity follows directly from \eqref{eq:logndef} by interchanging differentiation and integration, a manipulation which is justified by the dominated convergence theorem. In fact, by \cref{prop:ee} we have $\ee(v;\LN,\Lz) \in L^\infty(\R_+)$, which implies in turn 
$\big|\beta\,\partial_\beta\log(1-e^{-\beta v})\,e(v;\Lambda_N,\Lambda_0)\big| \leqslant \|e(\,\cdot\,;\Lambda_N,\Lambda_0)\|_\infty\, \frac{\beta_0 v}{e^{\beta_0 v}-1} \in L^1(\mathbb{R}_+)$ for any fixed $\beta_0 > 0$ and all $\beta \geqslant \beta_0$.

	\begin{lemma}\label{lemma:logetainf}
		Assume \eqref{eq:ass1} and \eqref{eq:ass2}. Then, as $\beta \to + \infty$,
			\begin{gather}
				\log \eta(\beta;\LN,\Lz) \sim - \frac{\pi^2}{\beta} \sum_{j = 0}^{+\infty} f_j\,\frac{(2\pi)^{2j}\,|B_{2j+2}|}{(j+1)(2j+1)}\,\frac{1}{\beta^{2j}}\,,\label{eq:logetainf}\\
				\partial_\beta \log \eta(\beta;\LN,\Lz) \sim \frac{\pi^2}{\beta^2} \sum_{j = 0}^{+\infty} f_j\,\frac{(2\pi)^{2j}\,|B_{2j+2}|}{(j+1)}\,\frac{1}{\beta^{2j}}\,,\label{eq:dlogetainf}
			\end{gather}
		where $B_{2j+2}$ are the Bernoulli numbers and the coefficients $\{f_{j}\}_{j \in \Nz} \subset \R$ are those appearing in \eqref{eq:asyev0}.				
	\end{lemma}
	
\begin{proof}
As an example, we discuss the derivation of \eqref{eq:logetainf}. The expansion \eqref{eq:dlogetainf} can be deduced by analogous arguments, starting from the identity \eqref{eq:dlogndef}.		

The key idea is that, in the low-temperature regime, the dominant contribution comes from the infrared region $v \to 0^+$. To make this intuition precise, we fix $\vz > 0$ small enough and split
	\begin{equation}\label{eq:logetaas}
			\log \eta(\beta;\LN,\Lz) = \int_0^{\vz} dv\, \log(1-e^{-\beta v})\,\ee(v;\LN,\Lz) 
			+ \int_{\vz}^\infty dv\, \log(1-e^{-\beta v})\,\ee(v;\LN,\Lz)  .
	\end{equation}
On $(0,\vz)$, we can use the expansion \eqref{eq:asyev0} truncated at an arbitrary order $J > 0$ and write
	\begin{multline}\label{eq:evIR}
		\int_0^{\vz} dv\, \log(1-e^{-\beta v})\,\ee(v;\LN,\Lz) 
		= \sum_{j = 0}^{J} f_j \int_0^{\infty} dv\, \log(1-e^{-\beta v}) \,v^{2j} \\
			- \sum_{j = 0}^{J} f_j \int_{\vz}^{\infty} dv\, \log(1-e^{-\beta v}) \,v^{2j}
			+ \int_0^{\vz} dv\, \log(1-e^{-\beta v})\,\bigg[\ee(v;\LN,\Lz)- \sum_{j = 0}^{J} f_j \,v^{2j}\bigg]\,.
	\end{multline}
The first sum determines the asymptotic expansion. Indeed, using the series $\log(1-x) = -\sum_{k = 1}^{+\infty} \frac{x^k}{k}$
and evaluating the resulting expressions in terms of the Euler Gamma and Riemann zeta functions \cite[Eqs. 5.2.1, 5.4.1 and 25.2.1, 25.6.2]{OL+10}, we obtain, for each $j \in \Nz$,
	\begin{multline*}
		\int_0^\infty dv\, \log(1-e^{-\beta v})\, v^{2j}
		= - \int_0^\infty dv\, \sum_{k = 1}^{+\infty} \frac{e^{-k \beta v}}{k}\,v^{2j} \\
		= - \frac{\Gamma(2j+1)}{\beta^{2j+1}} \sum_{k = 1}^{+\infty} \frac{1}{k^{2j+2}}
		= - \Gamma(2j+1)\,\zeta_{\mathrm{R}}(2j+2)\,\frac{1}{\beta^{2j+1}}
		= - \frac{(2\pi)^{2j+2}\,|B_{2j+2}|}{2 (2j+2)(2j+1)}\,\frac{1}{\beta^{2j+1}} \,.
	\end{multline*}
We now show that the remaining terms in \eqref{eq:evIR} yield corrections of larger order. Using the inequality $|\log(1-t)| \leqslant C\,t$ for $t \in [0,t_0]$ with $t_0 < 1$, together with some known identities for the incomplete Gamma function \cite[Eqs. 8.2.2 and 8.4.8]{OL+10}, for any $j \in \Nz$ and all sufficiently large $\beta > 0$ we infer
	\begin{multline}
		\left|\int_{\vz}^{\infty} dv\, \log(1-e^{-\beta v}) \,v^{2j}\right|
		\leqslant C \int_{\vz}^{\infty}\! dv\, e^{-\beta v} \,v^{2j}
		= C\, \beta^{-(2j+1)}\,\Gamma(2j+1,\beta \vz) \\
		= C\,\beta^{-(2j+1)}\,(2j)!\, e^{-\beta \vz} \sum_{k = 0}^{2j} \frac{(\beta \vz)^k}{k!}
		\leqslant C\,\beta^{-(2j+1)}\,(2j)!\, e^{-\beta \vz} (\beta \vz)^{2j} \sum_{k = 0}^{+\infty} \frac{1}{k!} = 
		C\,(2j)!\, \frac{\vz^{2j}}{\beta}\, e^{-(\beta \vz -1)}. \label{eq:estim}
	\end{multline}
For the last term in \eqref{eq:evIR}, since $|e(v;\Lambda_N,\Lambda_0) - \sum_{j=0}^J f_j v^{2j}| \leqslant C\,v^{2J+2}$ on $(0,v_0)$ by claim (i) in \cref{prop:ee}, computations similar to those described before entail
	\begin{multline*}
		\bigg|\int_0^{\vz}\!\! dv\, \log(1-e^{-\beta v}) \bigg[\ee(v;\LN,\Lz)- \tx\sum_{j = 0}^{J} f_j \,v^{2j}\bigg]\bigg| \\
		\leqslant C \int_0^{\infty}\!\! dv\, \big| \log(1-e^{-\beta v})\big|\,v^{2J+2}
		= C\, \frac{1}{\beta^{2J+3}}\,\Gamma(2J+3)\,\zeta_{\mathrm{R}}(2J+4)\,.
	\end{multline*}
	
Finally, exploiting the uniform boundedness of $\ee(v;\LN,\Lz)$ established in \cref{prop:ee} and using the bound \eqref{eq:estim}, we estimate the second integral in \eqref{eq:logetaas} by
	\begin{multline*}
			\left|\int_{\vz}^\infty dv\, \log(1-e^{-\beta v})\,\ee(v;\LN,\Lz)\right| \\
			\leqslant \|\ee(\,\cdot\,;\LN,\Lz)\|_{\infty} \int_{\vz}^\infty dv\, \log(1-e^{-\beta v}) 
			\leqslant  C\, \|\ee(\,\cdot\,;\LN,\Lz)\|_{\infty}\frac{1}{\beta}\, e^{-(\beta \vz - 1)} .
	\end{multline*}
We remark that the initial splitting at $\vz$ is essential. Indeed, since \eqref{eq:asyev0} holds only in the limit $v \to 0^+$, substituting it directly into \eqref{eq:logndef} would yield no uniform control of the remainder for large $v$.
\end{proof}

	\begin{lemma}\label{lemma:logeta0}
		Assume \eqref{eq:ass1} and \eqref{eq:ass2}. Then, as $\beta \to 0^+$,
			\begin{gather}
				\log \eta(\beta;\LN,\Lz) = \left(\int_0^\infty\! dv\,\ee(v;\LN,\Lz)\right) \log\beta + \int_0^\infty\! dv\,\log v\;\ee(v;\LN,\Lz) + \mathcal{O}(\beta \log\beta)\,, \label{eq:logeta0}\\
				\beta\,\partial_\beta \log \eta(\beta;\LN,\Lz) = \int_0^\infty\! dv\; \ee(v;\LN,\Lz) + \mathcal{O}(\beta \log\beta)\,.\label{eq:dlogeta0}
			\end{gather}		
	\end{lemma}
	
\begin{proof}
A formal expansion of the integrand functions in \eqref{eq:logndef} and \eqref{eq:dlogndef} suggests
\begin{gather}
	\log \eta(\beta;\LN,\Lz) = \int_0^\infty\!\!\! dv\, \big[\log(\beta v) + \mathcal{O}(\beta) \big]\,\ee(v;\LN,\Lz)\,, \\
	\beta\, \partial_\beta \log \eta(\beta;\LN,\Lz) = \int_0^\infty\!\!\! dv\; \big[1 + \mathcal{O}(\beta) \big]\,\ee(v;\LN,\Lz)\,.
\end{gather}
We now justify these expansions with quantitative bounds for the remainders.

On one hand, consider that $|\ee(v;\LN,\Lz)| \leqslant C/(1+v^2)$ for any $v > 0$ by \cref{prop:ee}. Then, using the elementary estimates $|\log(1-e^{-w}) - \log w| \leqslant C w$ for $0 < w <1$ and $|\log(1-e^{-w}) - \log w| \leqslant C(1 + \log w)$ for $w > 1$, for all $\beta > 0$ small enough we get
	\begin{multline}
		\bigg|\log \eta(\beta;\LN,\Lz) - \int_0^\infty\!\!\! dv\,\log(\beta v)\,\ee(v;\LN,\Lz)\bigg|
		\leqslant \int_0^\infty\!\!\! dv\, \big|\log(1-e^{-\beta v}) - \log(\beta v)\big|\,\big|\ee(v;\LN,\Lz)\big| \\
		\leqslant C \left( \beta\! \int_0^{1}\! dv\, v
			+ \beta\! \int_1^{1/\beta}\! dv\, \frac{1}{v} 
			+ \int_{1/\beta}^\infty\!\!\! dv\, \big[1+ \log(\beta v)\big]\,\frac{1}{v^2} \right) 
		= C \left( \tfrac{1}{2}\beta - \beta \log \beta + 2 \beta \right) 
		\leqslant C\, \beta\,|\log\beta|\,,
	\end{multline}
which proves \eqref{eq:logeta0}.

The asymptotic expansion \eqref{eq:dlogeta0} can be derived by similar arguments, observing that $\big|\frac{w}{e^{w}-1} - 1\big| \leqslant C\,w$ for $0 < w <1$ and $\big|\frac{w}{e^{w}-1} - 1\big| \leqslant C$ for $w > 1$.	
\end{proof}

We are now ready to prove \cref{cor:FUS}.

\begin{proof}[Proof of \cref{cor:FUS}]	
From the definitions \eqref{eq:Frdef}–\eqref{eq:Srdef}, combined with the identity \eqref{eq:ZZren} proved in \cref{thm:ZZren}, we can express the thermodynamic quantities as
	\begin{align*}
		\Fr(\beta) &= \mathcal{E}_{\vac}  + \tfrac{1}{\beta}\log\eta(\beta;\LN,\Lz)\,, \vspace{0.1cm} \\
		\Ur(\beta) &= \mathcal{E}_{\vac} + \partial_\beta \log\eta(\beta;\LN,\Lz)\,, \vspace{0.1cm} \\
		\Sr(\beta) &= \beta\,\partial_\beta \log\eta(\beta;\LN,\Lz) - \log\eta(\beta;\LN,\Lz)\,. 
	\end{align*}
Then, \eqref{eq:Frbinf}-\eqref{eq:Srb0} follow straightforwardly using the asymptotics established in \cref{lemma:logetainf,lemma:logeta0}. Absolute convergence of the integrals in \eqref{eq:Frb0}-\eqref{eq:Srb0} is ensured by the results established in \cref{prop:ee}.
\end{proof}

\subsection{Derivation of the Born series for the Casimir energy}\label{sec:bornexp}
We now turn to the proof of \cref{thm:EvacBorn}. To this end, we first provide an alternative representation of the relative spectral density $e(v;\Lambda_N,\Lambda_0)$.

	\begin{lemma}\label{lemma:eNalt}
		Assume \eqref{eq:ass1} and \eqref{eq:ass2}. Then, for all $v > 0$,
		\begin{equation}\label{eq:eexpalt}
		\ee(v;\LN,\Lz) =  \tx\sum_{j = 0}^{+\infty}\, \ee_j(v;\LN,\Lz)\,,
		\end{equation}
		where
		\begin{equation}\label{eq:Ejdef}
			\ee_j(v;\LN,\Lz) := \frac{1}{4 \pi^2}\,\Re \left(\sum_{\mm,\nn=1}^N \Big[ \Big(\big(\VNp \big)^{-1}\PNp\Big)^{j} (\VNp)^{-1} \Big]_{\mm\nn} \,e^{i v |x_\mm-x_\nn|}\right),
		\end{equation}
		and $\VNp,\PNp$ are as in \eqref{eq:VpPp}.
		Moreover, for any $j \in \Nz$ and $v > 0$,
		\begin{equation}\label{eq:Ejest}
			\big|\ee_j(v;\LN,\Lz)\big| \leqslant \frac{N}{4\pi^2 \min_\mm \alpha_\mm}
			\bigg(\sum_{\mm \neq \nn} \frac{1}{(4\pi \alpha_\mm)^2\, |x_\mm - x_\nn|^2}\bigg)^{j/2}
			\left[1 + \bigg(\frac{v}{4\pi \max_\mm \alpha_\mm}\bigg)^{\!2} \right]^{-\frac{j+1}{2}} .
		\end{equation}
	\end{lemma}
	
\begin{proof}
As a preliminary, recall the definition \eqref{eq:VNPN} of the matrices $\VN,\PN$ and observe that 
	\begin{equation*}
		\VNp = \lim_{\varepsilon \to 0^+} V_N^{v^2 e^{i\varepsilon}} , \qquad
		\PNp = \lim_{\varepsilon \to 0^+} P_N^{v^2 e^{i\varepsilon}} .
	\end{equation*}
Taking this into account and using \cref{lemma:GaN}, we can write the convergent Neumann series expansion
    \begin{equation*}
        \big[\Gamma_N^{+,v^2}\big]^{-1}
        = \sum_{j=0}^{+\infty}
        \Big[\big (\VNp\big)^{-1} \PNp \Big]^{j}
        \big(\VNp \big)^{-1}\,.
    \end{equation*}
Substituting this into the basic identity \eqref{eq:eexp} gives \eqref{eq:eexpalt}-\eqref{eq:Ejdef}.

To say more, by computations analogous to those reported in the proof of \cref{lemma:GaN}, we infer:
\begin{gather*}
	\big\|(\VNp)^{-1}\big\|^2 \leqslant \max_{\nn}\bigg(\,\frac{1}{\alpha_\mm^2 + \left(\frac{v}{4\pi}\right)^2}\bigg) \leqslant \frac{1}{(\min_\mm \alpha_\mm)^2}\; \frac{1}{1 + \left(\frac{v}{4\pi \max_\mm \alpha_\mm}\right)^2}\,; \\
	\big\|(\VNp)^{-1} \PNp\big\|^2_{\mathrm{HS}}  \leqslant \bigg(\sum_{\mm \neq \nn} \frac{1}{(4\pi \alpha_\mm)^2\, |x_\mm - x_\nn|^2}\bigg)\, \frac{1}{1+\left(\tfrac{v}{4\pi \max_\mm \alpha_\mm}\right)^2} \,.
\end{gather*}
Then, by Cauchy-Schwarz inequality and standard properties of Hilbert--Schmidt operators, we obtain
		\begin{multline*}
			\big|\ee_j(v;\LN,\Lz)\big|  
			\leqslant \frac{1}{4 \pi^2} \sum_{\mm,\nn=1}^N \left|\Big[ \Big(\big(\VNp\big)^{-1} \PNp\Big)^{j} (\VNp)^{-1} \Big]_{\mm\nn} \right| 
			\\
			\leqslant \frac{N}{4 \pi^2} \left\| \Big(\big(\VNp\big)^{-1} \PNp\Big)^{j} (\VNp)^{-1} \right\|_{\mathrm{HS}}
			\leqslant \frac{N}{4 \pi^2}\,\Big\|\big(\VNp\big)^{-1} \PNp\Big\|^{j}_{\mathrm{HS}}\, \big\|(\VNp)^{-1}\big\|
			\\
			\leqslant \frac{N}{4\pi^2 \min_\mm \alpha_\mm}\,
				\bigg(\sum_{\mm \neq \nn} \frac{1}{(4\pi \alpha_\mm)^2\, |x_\mm - x_\nn|^2}\bigg)^{j/2}
				\left[1 + \Big(\tfrac{v}{4\pi \max_\mm \alpha_\mm}\Big)^2 \right]^{\!-\frac{j+1}{2}}.
		\end{multline*}
which proves the bound \eqref{eq:Ejest}, thus concluding the proof.
\end{proof}		
		
For later reference, we record that direct calculations give
\begin{gather}
\ee_0(v;\LN,\Lz) = \sum_{\nn=1}^N \frac{4\alpha_\nn}{(4\pi \alpha_\nn)^2 + v^2}\,, \label{eq:E0exp} \\
\ee_1(v;\LN,\Lz) = \frac{1}{\pi}\,\Re\! \left(\sum_{\mm \neq \nn} \frac{e^{2i v\,|x_\mm-x_\nn|}}{|x_\mm-x_\nn|\,(4\pi \alpha_\mm - i v)(4\pi \alpha_\nn - i v)} \right). \label{eq:E1exp}
\end{gather}	

While \eqref{eq:E1exp} is easily seen to be consistent with the bound \eqref{eq:Ejest} for all $v > 0$, the explicit expression \eqref{eq:E0exp} shows that the bound is not sharp for $j = 0$ as $v \to +\infty$. This is due to an exact cancellation of the same nature as the one pointed out in \cref{rem:gj}. We also stress that \eqref{eq:Ejest}, together with the assumption \eqref{eq:ass2}, ensures uniform absolute convergence of the series \eqref{eq:eexpalt} for all $v > 0$.
\medskip

We now present an alternative characterization of the meromorphic extension of the relative zeta function in a neighbourhood of $s = -1/2$, of interest for the computation of the renormalized Casimir energy.

	\begin{lemma}\label{lemma:zetaAC2}
		Assume \eqref{eq:ass1} and \eqref{eq:ass2}. Then, for any $s \in \C$ with $- 1 < \Re s < 0$, the relative zeta function $\zeta(s;\LN,\Lz)$ can be expressed as		
		\begin{equation}
			\zeta(s;\LN,\Lz) = \zeta_0^{(\mathrm{AC})}(s;\LN,\Lz) + \zeta_1^{(\mathrm{AC})}(s;\LN,\Lz) + \sum_{j = 2}^{+\infty} \zeta_j(s;\LN,\Lz)\,,
		\end{equation}			
		where,
		\begin{equation}\label{eq:zeta0ren}
			\zeta_0^{(\mathrm{AC})}(s;\LN,\Lz) := \frac{1}{2 \cos(\pi s)} \sum_{\nn=1}^N\, (4\pi \alpha_\nn)^{-2s}\,,
		\end{equation}
		\begin{multline}\label{eq:zeta1ren}
			\zeta_1^{(\mathrm{AC})}(s;\LN,\Lz) :=
				 \frac{1}{2\pi} \sum_{\mm \neq \nn} \frac{1}{|x_\mm-x_\nn|^2}\int_0^{\infty} dv\;v^{-2s-1}\, 
				\bigg[ 2s\, \Im \bigg(\frac{e^{2i v\,|x_\mm-x_\nn|}}{(4\pi \alpha_\mm - i v)(4\pi \alpha_\nn - i v)}\bigg)
			\\
			- \Re \bigg( \frac{v (4\pi (\alpha_\mm + \alpha_\nn) - 2i v)\, e^{2i v\,|x_\mm-x_\nn|}}{(4\pi \alpha_\nn - i v)^2 (4\pi \alpha_\mm - i v)^2}\bigg) \bigg]\,,
		\end{multline}
		and
		\begin{equation}
			\zeta_j(s;\LN,\Lz) := \int_0^{\infty} dv\;v^{-2s}\, \ee_j(v;\LN,\Lz)\,, \qquad \mbox{for $j \geqslant 2$}\,,
		\end{equation}
		with $\ee_j(v;\LN,\Lz)$ as in \eqref{eq:Ejdef}. Moreover, $\zeta_1^{(\mathrm{AC})}(s;\LN,\Lz)$ and $\sum_{j = 2}^{+\infty} \zeta_j(s;\LN,\Lz)$ are analytic functions of $s$ in the strip $- 1 < \Re s < 0$.
	\end{lemma}

\begin{proof}
We start from the representation of the relative spectral density given in \cref{lemma:eNalt}. Taking into account \eqref{eq:ass2} together with the estimate \eqref{eq:Ejest}, we obtain
		\begin{multline}
			\big|\ee(v;\LN,\Lz)\big| \leqslant \sum_{j = 0}^{+\infty} \big|\ee_j(v;\LN,\Lz)\big| \\
			\leqslant \frac{N}{4\pi^2 \min_\mm \alpha_\mm}
				\sum_{j = 0}^{+\infty} \bigg(\sum_{\mm \neq \nn} \frac{1}{(4\pi \alpha_\mm)^2\, |x_\mm - x_\nn|^2}\bigg)^{j/2}
			\left[1 + \bigg(\frac{v}{4\pi \max_\mm \alpha_\mm}\bigg)^{\!2} \right]^{-\frac{j+1}{2}}\\
			\leqslant \frac{C}{\sqrt{(4\pi \max_\mm \alpha_\mm)^2 + v^2}}\,,
		\end{multline}
for some $C > 0$. This bound ensures that inserting \eqref{eq:eexpalt} into \eqref{eq:zetaev} yields an absolutely convergent integral for $0 < \Re s < 1/2$, defining an analytic function in that region, consistently with \eqref{eq:sinstrip}.

By Fubini's theorem, we may thus interchange summation and integration, and write
	\begin{equation*}
		\zeta(s;\LN,\Lz) = \sum_{j = 0}^{+\infty} \zeta_j(s;\LN,\Lz)\,,
	\end{equation*}
where, for $j \in \Nz$ and $0<\Re s < 1/2$, we have put
	\begin{equation*}
		\zeta_j(s;\LN,\Lz) := \int_0^{\infty} dv\;v^{-2s}\, \ee_j(v;\LN,\Lz)\,.
	\end{equation*}
In what follows, we treat the terms $j = 0$ and $j = 1$ explicitly and estimate the remainder. 

For $j = 0$, inserting the explicit expression \eqref{eq:E0exp}, changing variables and exploiting standard identities for the Euler beta function $\mathrm{B}(\cdot,\cdot)$ \cite[Eqs. (5.12.3), (5.12.1) and (5.5.3)]{OL+10}, we find, for $0 < \Re s < 1/2$,
	\begin{multline*}
		\zeta_0(s;\LN,\Lz) = \frac{1}{2\pi} \sum_{\nn=1}^N (4\pi \alpha_\nn)^{-2s} \int_0^{\infty} dw\; \frac{w^{-s-1/2}}{1 + w}\\
		= \frac{1}{2\pi} \sum_{\nn=1}^N (4\pi \alpha_\nn)^{-2s} \,\mathrm{B}\Big(\tfrac{1}{2}-s,\tfrac{1}{2}+s\Big)
		= \frac{1}{2 \cos(\pi s)} \sum_{\nn=1}^N (4\pi \alpha_\nn)^{-2s}\,.
	\end{multline*}
This yields \eqref{eq:zeta0ren} and provides the analytic continuation of $\zeta_0(s)$ to a meromorphic function with simple poles at all $s \in 1/2 + \Z$, consistently with \cref{thm:zLNLz}.
	
For $j = 1$, using \eqref{eq:E1exp} we get
	\begin{equation*}
		\zeta_1(s;\LN,\Lz) = \frac{1}{2\pi} \sum_{\mm \neq \nn} \frac{1}{|x_\mm-x_\nn|} \int_0^{\infty} dv\;v^{-2s} \left( \frac{e^{2i v\,|x_\mm-x_\nn|}}{(4\pi \alpha_\mm - i v)(4\pi \alpha_\nn - i v)} + \frac{e^{-2i v\,|x_\mm-x_\nn|}}{(4\pi \alpha_\mm + i v)(4\pi \alpha_\nn + i v)} \right) \,.
	\end{equation*}
We stress that this integral is absolutely convergent for $-1/2 < \Re s < 1/2$ and defines an analytic function of $s$ in the same strip. Notably, it cannot be directly employed to evaluate $\zeta_1(s)$ at $s = -1/2$. To reach this point, we integrate by parts the exponential expressions. Considering that boundary terms vanish for $-1/2 < \Re s < 0$, a direct computation yields
	\begin{align*}
		\zeta_1(s;\LN,\Lz) & = - \frac{1}{4\pi i} \sum_{\mm \neq \nn} \frac{1}{|x_\mm-x_\nn|^2} \bigg[
			\int_0^{\infty} dv\; e^{2i v\,|x_\mm-x_\nn|}\, \frac{d}{dv} \left(\frac{v^{-2s}}{(4\pi \alpha_\mm - i v)(4\pi \alpha_\nn - i v)} \right) \\
			& \hspace{5cm} - \int_0^{\infty} dv\; e^{-2i v\,|x_\mm-x_\nn|}\, \frac{d}{dv} \left(\frac{v^{-2s}}{(4\pi \alpha_\mm + i v)(4\pi \alpha_\nn + i v)}\right)
			\bigg] \\
		& = \frac{1}{2\pi} \sum_{\mm \neq \nn} \frac{1}{|x_\mm-x_\nn|^2}\int_0^{\infty} dv\;v^{-2s-1}\, 
				\bigg[ 2s\, \Im \bigg(\frac{e^{2i v\,|x_\mm-x_\nn|}}{(4\pi \alpha_\mm - i v)(4\pi \alpha_\nn - i v)}\bigg)
			\\
			& \hspace{7cm} - \Im \bigg( \frac{i v (4\pi (\alpha_\mm + \alpha_\nn) - 2i v)\, e^{2i v\,|x_\mm-x_\nn|}}{(4\pi \alpha_\nn - i v)^2 (4\pi \alpha_\mm - i v)^2}\bigg) \bigg]\,.
	\end{align*}
	
The last expression coincides with \eqref{eq:zeta1ren} and it is easy to see that the integral converges absolutely for $-1 < \Re s < 0$, thus providing the desired analytic continuation and proving the claimed analyticity of $\zeta_1^{(\mathrm{AC})}(s;\LN,\Lz)$.	

As for the terms with $j \geqslant 2$, using again \eqref{eq:Ejest} and relabelling the summation index, we infer
	\begin{multline*}
		\bigg| \sum_{j = 2}^{+\infty} \zeta_j(s;\LN,\Lz) \bigg| 
		\leqslant \frac{N}{4\pi^2 \min_\mm \alpha_\mm} \sum_{j' = 0}^{+\infty}\bigg(\sum_{\mm \neq \nn} \frac{1}{(4\pi \alpha_\mm)^2\, |x_\mm - x_\nn|^2}\bigg)^{j'/2 + 1}\;\times \\
		\times\,
		\int_0^{\infty}\!\! dv\;v^{-2 \Re s} \left[1 + \Big(\tfrac{v}{4\pi \max_\mm \alpha_\mm}\Big)^2 \right]^{-3/2} .		
	\end{multline*}
Then, by arguments analogous to those outlined in the proof of \cref{lemma:ZIR}, we infer that $\sum_{j = 2}^{+\infty} \zeta_j(s;\LN,\Lz)$ is an analytic function of $s \in \C$ for $-1 < \Re s < 1/2$, thereby concluding the proof.	
\end{proof} 
 
We are ready to establish \cref{thm:EvacBorn}.

\begin{proof}[Proof of \cref{thm:EvacBorn}]
From \eqref{eq:ZZzeta} and the results established in \cref{thm:ZZren}, it follows that the renormalized Casimir energy can be written as
\begin{equation}\label{eq:Evac0}
\mathcal{E}_{\vac} = \big[1-\log(2\ell) \big] \Reu_{s = -1/2} \zeta(s;\LN,\Lz) + \tfrac{1}{2} \Rez_{s = -1/2} \zeta(s;\LN,\Lz) \,.
\end{equation}
By \cref{lemma:zetaAC2}, the residue of $\zeta(s;\LN,\Lz)$ at $s = -1/2$ is entirely determined by the corresponding residue of $\zeta_0^{(\mathrm{AC})}(s)$. More precisely, recalling \eqref{eq:zeta0ren}, we obtain
\begin{equation*}
\Reu_{s=-1/2} \zeta(s;\LN,\Lz) 
= \Reu_{s=-1/2} \zeta_0^{(\mathrm{AC})}(s;\LN,\Lz)
= 2\tx\sum_{\nn=1}^N \alpha_\nn\,.
\end{equation*}
Concerning the finite parts, direct calculations yield
\begin{equation*}
\Rez_{s=-1/2} \zeta_0^{(\mathrm{AC})}(s;\LN,\Lz)
= -4 \tx\sum_{\nn=1}^N \alpha_\nn\, \log(4 \pi \alpha_\nn)
\end{equation*}
\begin{multline*}
	\Rez_{s = -1/2} \zeta_1^{(\mathrm{AC})}(s;\LN,\Lz) 
		=  - \frac{1}{2\pi} \sum_{\mm \neq \nn} \frac{1}{|x_\mm-x_\nn|^2}\int_0^{\infty}\!\! dv\; \bigg[
				\Im \bigg(\frac{e^{2i v\,|x_\mm-x_\nn|}}{(4\pi \alpha_\mm - i v)(4\pi \alpha_\nn - i v)}\bigg)
			\\
			 + \Re \bigg(\frac{v (4\pi (\alpha_\mm + \alpha_\nn) - 2i v)\, e^{2i v\,|x_\mm-x_\nn|}}{(4\pi \alpha_\nn - i v)^2 (4\pi \alpha_\mm - i v)^2}\bigg) \bigg] \,,
\end{multline*}

\begin{equation*}
		\Rez_{s = -1/2}\zeta_j(s) = \int_0^{\infty} dv\;v\, \ee_j(v;\LN,\Lz)\,, \qquad \mbox{for $j \geqslant 2$}\,.
\end{equation*}
Substituting the above expressions into \eqref{eq:Evac0} and recalling the definition \eqref{eq:Ejdef} of $e_j(v;\LN,\Lz)$ yields \eqref{eq:Eren2}-\eqref{eq:Ejren}.

Let us finally establish \eqref{eq:ErenRest}, which in particular implies the absolute convergence of the series in \eqref{eq:Eren2}. To this purpose, we fix $\rho$ as in \eqref{eq:rho} and notice that the assumption \eqref{eq:ass2} actually ensures $\rho \in (0,1)$.
Then, using the bound \eqref{eq:Ejest}, for any $J \geqslant 2$ we obtain
\begin{multline*}
		\sum_{j = J}^{+\infty} \big|\mathcal{E}_{j}\big|
		\leqslant \int_0^{\infty} dv\;v\, \big|\ee_j(v;\LN,\Lz)\big|\\
		\leqslant \frac{N}{4\pi^2 \min_\mm \alpha_\mm} \sum_{j = J}^{+\infty} \,\rho^{j} \int_0^{\infty}\! dv\;v \left[1 + \Big(\frac{v}{4\pi \max_\mm \alpha_\mm}\Big)^2 \right]^{-\frac{j+1}{2}} 
		= \frac{4N\,(\max_\mm \alpha_\mm)^2}{\min_\mm \alpha_\mm} \sum_{j = J}^{+\infty} \frac{\rho^{j}}{j-1}\,.
\end{multline*}
Concerning the series in the latter expression, we have the elementary chain of identities
\begin{equation*}
\sum_{j = J}^{+\infty} \frac{\rho^j}{j-1}
= \sum_{j = J}^{+\infty} \int_0^{\rho} w^{j-2}\,dw
= \int_0^{\rho} w^{J-2} \sum_{k=0}^{\infty} w^k\,dw
= \int_0^{\rho} \frac{w^{J-2}}{1-w}\,dw ,
\end{equation*}
where the exchange of summation and integration is justified by dominated convergence, since the geometric series is uniformly convergent on every compact subset of $(0,1)$. Considering this and using the elementary estimates
	\begin{align*}
		\int_0^{\rho} dw\; \frac{w^{J-2}}{1-w} 
			& \leqslant \frac{1}{1-\rho} \int_0^{\rho} dw\;w^{J-2} = \frac{\rho^{J-1}}{(J-1)(1-\rho)}\,; \\
		\int_0^{\rho} dw\; \frac{w^{J-2}}{1-w} 
			& \leqslant \rho^{J-2} \int_0^{\rho} dw\; \frac{1}{1-w} = - \rho^{J-2}\,\log(1-\rho)\,. 
	\end{align*}
yields \eqref{eq:ErenRest}, therefore concluding the proof.
\end{proof}

We now specialize the results of \cref{thm:EvacBorn} to the case of identical obstacles. Therefore, on top of the basic hypotheses \eqref{eq:ass1} and \eqref{eq:ass2}, we henceforth assume \eqref{eq:aannaa} and work with the rescaled coordinates introduced in \eqref{eq:yn}, namely,
\begin{equation*}
	\alpha_\nn = \alpha > 0\,, \qquad y_\nn = 4\pi \alpha\,x_\nn\,, \qquad \mbox{for all $\nn \in \{1,\dots,N\}$}\,.
\end{equation*}
	
\begin{proof}[Proof of \cref{cor:num}]
We appeal to the representation established in \cref{thm:EvacBorn} and proceed to examine separately the contributions $\mathcal{E}_{0}^{\ren},\mathcal{E}_{1}^{\ren}$ and $\mathcal{E}_{j}$, for $j \geqslant 2$, introduced therein.

First, it is easy to see that \eqref{eq:E0ren} reduces to
	\begin{equation}\label{eq:E0num}
		\mathcal{E}_{0}^{\ren} = 2 N \alpha\, \big[1-\log(8\pi \alpha \ell) \big] \,.
	\end{equation}

Next, we consider $\mathcal{E}_{1}^{\ren}$. We start from \eqref{eq:E1ren} and perform the natural change of integration variable $w := v/(4\pi \alpha) \in (0,\infty)$. After straightforward manipulations, this yields
	\begin{multline}\label{eq:E1num}
				\mathcal{E}_{1}^{\ren} = - \alpha \sum_{\mm \neq \nn} \frac{1}{|y_\mm-y_\nn|^2}\;\Im \left(\int_0^{\infty}\!\! dw\; \frac{1 + w^2}{(1 - i w)^4} \;e^{2i w\,|y_\mm-y_\nn|}\right)\\
				= \alpha \sum_{\mm \neq \nn} \frac{1}{|y_\mm-y_\nn|^2}
				\left[
					\Im \left(\int_0^{\infty}\!\! dw\; \frac{e^{2i w\,|y_\mm-y_\nn|}}{(1 - i w)^2} \right)
					-2\, \Im \left(\int_0^{\infty}\!\! dw\;\frac{e^{2i w\,|y_\mm-y_\nn|}}{(1 - i w)^3}\right)				 
				\right] .
	\end{multline}

The expression \eqref{eq:Ejren} for $\mathcal{E}_{j}$, with $j \geqslant 2$, can be handled similarly. Since the matrix $\VNp$ defined in \eqref{eq:VpPp} is a scalar multiple of the identity in the present setting, a direct calculation using the explicit form of $\PNp$ gives
	\begin{multline}\label{eq:Ejnum}
		\mathcal{E}_{j} = 2 \alpha \sum_{\mm \neq p_1 \neq \dots \neq p_{j-1} \neq \nn} \frac{1}{|y_\mm - y_{p_1}| \cdot |y_{p_1} - y_{p_2}| \cdot \dots \cdot |y_{p_{j-1}} - y_{\nn}|} \;\times \\
				\times\; \Re \left( \int_0^{\infty}\!\! dw\;\frac{w}{(1-i w)^{j+1}}\; e^{i w\,\big( |y_\mm - y_{p_1}| + |y_{p_1} - y_{p_2}| +\, \dots + |y_{p_{j-1}} - y_{\nn}| + |y_{\nn}-y_{\mm}|\big)} \right) \\
		= 2 \alpha \sum_{\mm \neq p_1 \neq \dots \neq p_{j-1} \neq \nn} \frac{1}{|y_\mm - y_{p_1}| \cdot |y_{p_1} - y_{p_2}| \cdot \dots \cdot |y_{p_{j-1}} - y_{\nn}|} \;\times \\
				\times\; \bigg[ 
					\Im \left(\int_0^{\infty}\!\! dw\;\frac{1}{(1-i w)^{j+1}}\; e^{i w\,\big( |y_\mm - y_{p_1}| + |y_{p_1} - y_{p_2}| +\, \dots + |y_{p_{j-1}} - y_{\nn}| + |y_{\nn}-y_{\mm}|\big)} \right) \\
					- \Im \left(\int_0^{\infty}\!\! dw\;\frac{1}{(1-i w)^{j}}\; e^{i w\,\big( |y_\mm - y_{p_1}| + |y_{p_1} - y_{p_2}| +\, \dots + |y_{p_{j-1}} - y_{\nn}| + |y_{\nn}-y_{\mm}|\big)}\right)
				\bigg]\,,
			\end{multline}
where the sum runs over all ordered $(j+1)$-tuples $(\mm,p_1,\dots,p_{j-1},\nn)$ of indexes in $\{1,\dots,N\}$ such that no two consecutive ones are allowed to be equal.

Recalling the position \eqref{eq:Etilde}, the above results yield \eqref{eq:Erennum}-\eqref{eq:Ejnumc} once we establish that
	\begin{equation*}
		\Im\left(\int_0^{\infty} dw\;\frac{e^{i w r}}{(1-i w)^k}\right) 
		= \frac{1}{(k-1)!} \left[\, \sum_{h = 0}^{k-2} (k-h-2)!\,(-r)^h + (- r)^{k-1}\,e^{r} E_1(r)\,\right] ,
	\end{equation*}
for all $r > 0$ and integer $k \geqslant 2$. This identity follows readily from the arguments reported in Appendix \ref{sec:app}.

Finally, the remainder estimate \eqref{eq:ErenRestnum} is a trivial rephrasing of \eqref{eq:ErenRest} and \eqref{eq:rho}.
\end{proof}
	
\appendix
\section{Auxiliary integral identities}\label{sec:app}
This appendix is devoted to the evaluation of a specific class of integrals arising in the computation of the renormalised Casimir energy for identical point-like obstacles.

For any $k \geqslant 1$ and $r > 0$, let us define
	\begin{equation}\label{eq:Ij}
		\II_{k}(r) := \int_0^{\infty} dw\;\frac{e^{i w r}}{(1-i w)^k}\,.
	\end{equation}
While the integral converges absolutely for $k \geqslant 2$, it is only conditionally convergent for $k = 1$. Yet, in the latter case a standard identity for the exponential integral $E_1(\cdot)$ gives the following, see \cite[6.7.1]{OL+10}:
	\begin{equation}\label{eq:I1}
		\II_{1}(r) = i \int_0^{\infty} dw\;\frac{e^{i w r}}{w + i} = i\,e^{r} E_1(r)\,.
	\end{equation}

\begin{lemma}
	For any $r > 0$ and integer $k \geqslant 2$, the integral $\II_k(r)$ can be expressed as
	\begin{equation}\label{eq:IjI1}
	\II_{k}(r) = \frac{i}{(k-1)!}\left[ \;\sum_{h = 0}^{k-2}\, (k-h-2)!\,(-r)^h + (- r)^{k-1}\,e^{r} E_1(r) \right].
	\end{equation}
\end{lemma}

\begin{proof}
An elementary integration by parts yields, for any fixed $k \geqslant 2$ and $r > 0$,
	\begin{multline*}
		\II_{k}(r) 
		= \int_0^{\infty}\!\! dw\;\frac{d}{dw}\left(\frac{1}{i(k-1)(1-i w)^{k-1}}\right) e^{i w r} \\
		= \frac{1}{i(k-1)} \left[\left.\frac{1}{(1-i w)^{k-1}}\,e^{i w r} \right|_0^{\infty} - i r \int_0^{\infty}\!\! dw\;\frac{1}{(1-i w)^{k-1}}\; e^{i w r} \right]
		= \frac{1}{i(k-1)}\, \Big[- 1 - i r\, \II_{k-1}(r) \Big]\,,
	\end{multline*}
which yields
	\begin{equation}\label{eq:IIjind}
	\II_{k}(r) = \frac{i}{k-1} - \frac{r}{k-1}\, \II_{k-1}(r) \,.
	\end{equation}
Starting from here, one infers \eqref{eq:IjI1} as an educated guess. Its validity can be checked a posteriori by a standard inductive argument, using \eqref{eq:IIjind}.
\end{proof}
\bigskip

\section*{Acknowledgments}
The numerical computations and plots were generated using the software \texttt{Mathematica} (Wolfram Research Inc., Version 14.3.0, 2025, Champaign, IL).

\section*{Funding}
The present research has been supported by
{\sl MUR} grant ``Dipartimento di Eccellenza'' 2023-2027 of Dipartimento di Matematica, Politecnico di Milano
and
{\sl PRIN 2022} grant ``ONES - OpeN and Effective Quantum Systems'' (prot. 2022L45WA3).
D.F. is member of Gruppo Nazionale per la Fisica Matematica GNFM – INdAM in Italy. 
M.G. acknowledges the support of {\sl PNRR Italia Domani and Next Generation EU - ICSC National Research Centre for High Performance Computing, Big Data and Quantum Computing}.

\section*{Conflict of interest}
The authors declare no conflict of interest.

\section*{Data availability statement}
Data sharing is not applicable to this article as it has no associated data.


\bigskip

\begin{thebibliography}{999}

\bibitem[ACS16]{ACS16} S. Albeverio, C. Cacciapuoti, M. Spreafico, {\sl Relative partition function of Coulomb plus delta interaction}, pp. 1-29 in J. Dittrich, H. Kovarik, A. Laptev (Eds.) ``Functional Analysis and Operator Theory for Quantum Physics. A Festschrift in Honor of Pavel Exner'', European Mathematical Society Publishing House, Z\"urich (2016).

\bibitem[AC+10]{AC+10} S. Albeverio, G. Cognola, M. Spreafico, S. Zerbini, {\sl Singular perturbations with boundary conditions and the Casimir effect in the half space}, J. Math. Phys. {\bf 51}, 063502 (2010).

\bibitem[AG+88]{AG+88} S. Albeverio, F. Gesztesy, R. Hoegh-Krohn, H. Holden, {\sl Solvable Models in Quantum Mechanics}, Springer Berlin, Heidelberg (1988).

\bibitem[AK99]{AK99} S. Albeverio, P. Kurasov, {\sl Singular Perturbations of Differential Operators}, London Mathematical Society Lecture Notes Series 271, Cambridge University Press, Cambridge (1999).

\bibitem[B14]{B14} R. Bennett, {\sl Born-series approach to the calculation of Casimir forces}, Phys. Rev. A {\bf 89}, 062512 (2014).

\bibitem[BP34]{BP34} H.A. Bethe, R. Peierls, {\sl Quantum Theory of the Diplon}, Proc. R. Soc. A {\bf 148}, 146-156 (1934).

\bibitem[BK+09]{BK+09} M. Bordag, G.L. Klimchitskaya, U. Mohideen, V.M. Mostepanenko, {\sl Advances in the Casimir Effect}, Oxford University Press Inc., New York (2009).

\bibitem[BC+02]{BC+02} G. Bressi, G. Carugno, R. Onofrio, G. Ruoso, {\sl Measurement of the Casimir Force between Parallel Metallic Surfaces}, Phys. Rev. Lett. {\bf 88}(4), 041804 (2002).

\bibitem[BC+03]{BC+03} A.A. Bytsenko, G. Cognola, E. Elizalde, V. Moretti, S. Zerbini, {\sl Analytic aspects of quantum fields}, World Scientific Publishing Co., Singapore (2003).

\bibitem[CFP17]{CFP17} C. Cacciapuoti, D. Fermi, A. Posilicano, {\sl Relative-Zeta and Casimir energy for a semitransparent hyperplane selecting transverse modes}, pp. 71-97 in G.F. Dell'Antonio, A. Michelangeli (Eds.), ``Advances in Quantum Mechanics: contemporary trends and open problems'', Springer (2017).

\bibitem[CCF26]{CCF26} D. Cafiero, M. Correggi, D. Fermi, {\sl Homogenization of point interactions}, arXiv:2603.21400 [math-ph].

\bibitem[C48]{C48} H.B.G. Casimir, {\sl On the attraction between two perfectly conducting plates}, Proc. Koninklijke Nederlandse Akad. van Wetenschappen {\bf 51}, 793-795 (1948).

\bibitem[CP48]{CP48} H.B.G. Casimir, D. Polder, {\sl The Influence of Retardation on the London-van der Waals Forces}, Phys. Rev. {\bf 73}, 360-372 (1948).

\bibitem[DM+11]{DM+11} D. Dalvit, P. Milonni, D. Roberts, F. Da Rosa (Eds.), {\sl Casimir Physics}, Lecture Notes in Physics {\bf 834}, Springer-Verlag Berlin Heidelberg (2011).

\bibitem[DM24]{DM24} M. Dhital, U. Mohideen, {\sl A Brief Review of Some Recent Precision Casimir Force Measurements}, Physics {\bf 6}(2), 891-904 (2024).

\bibitem[DL+14]{DL+14} M. Dou, F. Lou, M. Bostr\"om, I. Brevik, C. Persson, {\sl Casimir quantum levitation tuned by means of material properties and geometries}, Phys. Rev. B {\bf 89}, 201407 (2014).

\bibitem[DC76]{DC76} J.S. Dowker, R. Critchley, {\sl Effective Lagrangian and energy-momentum tensor in de Sitter space}, Phys. Rev. D {\bf 13}(12), 3224-3232 (1976).

\bibitem[DK78]{DK78} J.S. Dowker, G. Kennedy, {\sl Finite temperature and boundary effects in static space-times}, J. Phys. A {\bf 11}(5), 895-920 (1978).


\bibitem[EO+94]{EO+94} E. Elizalde, S.D. Odintsov, A. Romeo, A.A. Bytsenko, S. Zerbini, {\sl Zeta regularization techniques with applications}, World Scientific Publishing Co., Singapore (1994).

\bibitem[EY+24]{EY+24} B. Elsaka, X. Yang, P. K\"astner, K. Dingel, B. Sick, P. Lehmann, S.Y. Buhmann, H. Hillmer, {\sl Casimir Effect in MEMS: Materials, Geometries, and Metrologies - A Review}, Materials {\bf 17}(14), 3393 (2024).

\bibitem[EB23]{EB23} T. Emig, G. Bimonte, {\sl Multiple Scattering Expansion for Dielectric Media: Casimir Effect}, Phys. Rev. Lett. {\bf 130}, 200401 (2023).

\bibitem[FS22]{FS22} Y.L. Fang, A. Strohmaier, {\sl A Mathematical Analysis of Casimir Interactions I: The Scalar Field}, Ann. Henri Poincaré {\bf 23}, 1399-1449 (2022).

\bibitem[F36]{F36} E. Fermi, {\sl Sul moto dei neutroni nelle sostanze idrogenate}, Ricerca Scientifica {\bf 7}(2), 13-53 (1936).

\bibitem[F16]{F16} D. Fermi, {\sl A functional analytic framework for local zeta regularization and the scalar Casimir effect}, PhD thesis, Doctoral School in Mathematics, 28\textsuperscript{th} Cycle, University of Milan (2016). \href{https://doi.org/10.13130/d-fermi phd2016-02-22}{https://doi.org/10.13130/d-fermi phd2016-02-22}.

\bibitem[F20]{F20} D. Fermi, {\sl The Casimir energy anomaly for a point interaction}, Mod. Phys. Lett. A {\bf 35}(03), 2040008 (2020).

\bibitem[FP11]{FP11} D. Fermi, L. Pizzocchero, {\sl Local zeta regularization and the Casimir effect}, Prog. Theor. Phys. {\bf 126}(3), 419-434 (2011).

\bibitem[FP15]{FP15} D. Fermi, L. Pizzocchero, {\sl Local zeta regularization and the scalar Casimir effect III. The case with a background harmonic potential}, Int. J. Mod. Phys. A {\bf 30}(35), 1550213 (2015).

\bibitem[FP16]{FP16} D. Fermi, L. Pizzocchero, {\sl Local zeta regularization and the scalar Casimir effect IV. The case of a rectangular box}, Int. J. Mod. Phys. A {\bf 31}(4-5), 1650003 (2016).

\bibitem[FP17]{FP17} D. Fermi, L., Pizzocchero, {\sl Local zeta regularization and the scalar Casimir effect. A general approach based on integral kernels}, World Scientific, Singapore (2017).

\bibitem[FP18]{FP18} D. Fermi, L. Pizzocchero, {\sl Local Casimir effect for a scalar field in presence of a point impurity}, Symmetry {\bf 10}(2), 38 (2018).

\bibitem[FP23]{FP23} D. Fermi, L. Pizzocchero, {\sl On the Casimir Effect with $\delta$-Like Potentials, and a Recent Paper by K. Ziemian (Ann. Henri Poincar\'e, 2021)}, Ann. Henri Poincar\'e {\bf 24}, 2363-2400 (2023).

\bibitem[G77]{G77} G. Gibbons, {\sl Thermal zeta functions}, Phys. Lett. A {\bf 60}(5), 385-386 (1977).

\bibitem[GS17]{GS17} Y.V. Grats, P. Spirin, {\sl Vacuum polarization and classical self-action near higher-dimensional defects}, Eur. Phys. J. C {\bf 77}, 101 (2017).

\bibitem[GS25]{GS25} Y.V. Grats, P. Spirin, {\sl Vacuum Polarization Effects of Pointlike Singularities}, Moscow Univ. Phys. {\bf 80} (Suppl 2), S564-S571 (2025).

\bibitem[GK97]{GK97} R.E. Green, S.G. Krantz, {\sl Function theory of one complex variable}, John Wiley \& Sons Inc., New York (1997). 

\bibitem[H77]{H77} S.W. Hawking, {\sl Zeta function regularization of path integrals in curved spacetime}, Commun. Math. Phys. {\bf 55}(2), 133-148 (1977).

\bibitem[H05]{H05} A. Herdegen, {\sl Quantum backreaction (Casimir) effect I. What are admissible idealizations?}, Ann. Henri Poincar\'e {\bf 6}, 657-695 (2005).

\bibitem[H06]{H06} A. Herdegen, {\sl Quantum backreaction (Casimir) effect II. Scalar and electromagnetic fields}, Ann. Henri Poincar\'e {\bf 7}, 253-301 (2006).

\bibitem[HS10]{HS10} A. Herdegen, M. Stopa, {\sl Global versus local Casimir effect}, Ann. Henri Poincar\'e {\bf 11}, 1171-1200 (2010).

\bibitem[KC+07]{KC+07} D.J. Kapner, T.S. Cook, E.G. Adelberger, J.H. Gundlach, B.R. Heckel, C.D. Hoyle, H.E. Swanson, {\sl Tests of the Gravitational Inverse-Square Law below the Dark-Energy Length Scale}, Phys. Rev. Lett. {\bf 98}, 021101 (2007).

\bibitem[KG89]{KG89} J.I. Kapusta, C. Gales, {\sl Finite-Temperature Field Theory, Principles and Applications}, Cambridge University Press, Cambridge (1989).

\bibitem[K01]{K01} K. Kirsten, {\sl Spectral Functions in Mathematics and Physics}, CRC Press, Boca Raton, Florida (2001).

\bibitem[L97]{L97} S.K. Lamoreaux, {\sl Demonstration of the Casimir Force in the $0.6$ to $6 \mu m$ Range}, Phys. Rev. Lett. {\bf 78}(1), 8pp. (1997). Erratum Phys. Rev. Lett. {\bf 81}(24), 5475-5476 (1998).

\bibitem[L05]{L05} S.K. Lamoreaux, {\sl The Casimir force: background, experiments, and applications}, Rep. Prog. Phys. {\bf 68}, 201-236 (2005).

\bibitem[L55]{L55} E.M. Lifshitz, {\sl The Theory of Molecular Attractive Forces between Solids}, Sov. Phys. JETP {\bf 2}(1), 73-83 (1956). See also Zh. Eksp. Teor. Fiz. {\bf 29}, 94 (1955) [in Russian].

\bibitem[M94]{M94} P.W. Milonni, {\sl The Quantum Vacuum, An Introduction to Quantum Electrodynamics}, Academic Press Inc., San Diego (1994).

\bibitem[M01]{M01} K. A. Milton, {\sl The Casimir effect - Physical manifestations of zero-point energy}, World Scientific Publishing Co., Singapore (2001).

\bibitem[MR98]{MR98} U. Mohideen, A. Roy, {\sl Precision Measurement of the Casimir Force from $0.1$ to $0.9\mu m$}, Phys. Rev. Lett. {\bf 81}(21), 4549-4552 (1998).

\bibitem[MT97]{MT97} V.M. Mostepanenko, N.N. Trunov, {\sl The Casimir effect and its applications}, Oxford University Press Inc., New York (1997).

\bibitem[M98]{M98} W. M\"uller, {\sl Relative Zeta Functions, Relative Determinants and Scattering Theory}, Commun. Math. Phys. {\bf 192}, 309–347 (1998).

\bibitem[NF+13]{NF+13} V.D. Nguyen, S. Faber, Z. Hu, G.H. Wegdam, P. Schall, {\sl Controlling colloidal phase transitions with critical Casimir forces}, Nature Commun. {\bf 4}, 1584 (2013).

\bibitem[OL+10]{OL+10} F.W.J. Olver, D.W. Lozier, R.F. Boisvert, C.W. Clark, {\sl NIST Handbook of mathematical functions}, Cambridge University Press, Cambridge (2010).

\bibitem[PR+25]{PR+25} R. Passante, L. Rizzuto, P. Schall, E. Marino, {\sl Quantum and critical Casimir effects: bridging fluctuation physics and nanotechnology}, Nanoscale {\bf 17}, 13982-13997 (2025).

\bibitem[RE+09]{RE+09} S.J. Rahi, T. Emig, N. Graham, R.L. Jaffe, M. Kardar, {\sl Scattering Theory Approach to Electrodynamic Casimir Forces}, Phys. Rev. D{\bf 80}, 085021 (2009).

\bibitem[RS74]{RS74} D.B. Ray, I.M. Singer, {\sl R-torsion and the Laplacian on Riemannian manifolds}, Adv. Math. {\bf 7}, 145-210 (1974).

\bibitem[RCD+10]{RCD+10} S. Reynaud, A. Canaguier-Durand, R. Messina, A. Lambrecht, P.A. Maia Neto, {\sl The Scattering Approach to the Casimir Force}, Int. J. Mod. Phys. A {\bf 25}(11), 2201-2211 (2010).

\bibitem[S05]{S05} A. Scardicchio, {\sl Casimir dynamics: interactions of surfaces with codimension $> 1$ due to quantum fluctuations}, Phys. Rev. D {\bf 72}, 065004 (2005).

\bibitem[SC+23]{SC+23} F. Schmidt, A. Callegari, A. Daddi-Moussa-Ider, B. Munkhbat, R. Verre, T. Shegai, M. K\"all, H. L\"owen, A. Gambassi, G. Volpe, {\sl Tunable critical Casimir forces counteract Casimir–Lifshitz attraction}, Nature Physics {\bf 19}, 271–278 (2023).

\bibitem[SZ09]{SZ09} M. Spreafico, S. Zerbini, {\sl Finite temperature quantum field theory on noncompact domains and application to delta interactions}, Rep. Math. Phys. {\bf 63}(2), 163-177 (2009).

\bibitem[S21]{S21} A. Strohmaier, {\sl The Classical and Quantum Photon Field for Non-compact Manifolds with Boundary and in Possibly Inhomogeneous Media}, Commun. Math. Phys. {\bf 387}, 1441-1489 (2021).

\bibitem[S24]{S24} A. Strohmaier, {\sl Dimensional reduction formulae for spectral traces and Casimir energies}, Lett. Math. Phys. {\bf 114}, 66 (2024).

\bibitem[TP25]{TP25} F. Tajik, G. Palasantzas, {\sl Dynamical actuation of graphene MEMS under the influence of Casimir forces}, EPL {\bf 150}, 36001 (2025).

\bibitem[W79]{W79} R.M. Wald, {\sl On the Euclidean approach to quantum field theory in curved spacetime}, Commun. Math. Phys. {\bf 70}(3), 221–242 (1979).

\bibitem[Z21]{Z21} K. Ziemian, {\sl Algebraic approach to Casimir force between two $\delta$-like potentials}, Ann. Henri Poincar\'e {\bf 22}, 1751-1781 (2021).

\end{thebibliography}
\end{document}